\documentclass[11pt]{article}
\usepackage[utf8]{inputenc}
\usepackage{style}
\usepackage{framed}
\usepackage[parfill]{parskip}
\usepackage{setspace}
\usepackage{hyperref}
\allowdisplaybreaks

\title{ETH-Hardness of Approximating 2-CSPs and \\ Directed Steiner Network}

\author{
Irit Dinur\thanks{Email: \texttt{irit.dinur@weizmann.ac.il}. This work was supported by BSF grant 2014371.} \vspace{-0.5em}\\
Weizmann Institute of Science
\and
Pasin Manurangsi\thanks{Email: \texttt{pasin@berkeley.edu}. This work was done while the author was visiting Weizmann Institute of Science.} \vspace{-0.5em}\\
UC Berkeley
}

\begin{document}

\maketitle
\thispagestyle{empty}

\begin{abstract}
We study the 2-ary constraint satisfaction problems (2-CSPs), which can be stated as follows: given a constraint graph $G = (V, E)$, an alphabet set $\Sigma$ and, for each edge $\{u, v\} \in E$, a constraint $C_{uv} \subseteq \Sigma \times \Sigma$, the goal is to find an assignment $\sigma: V \to \Sigma$ that satisfies as many constraints as possible, where a constraint $C_{uv}$ is said to be satisfied by $\sigma$ if $(\sigma(u), \sigma(v)) \in C_{uv}$.

While the approximability of 2-CSPs is quite well understood when the alphabet size $|\Sigma|$ is constant (see e.g.~\cite{Rag08}), many problems are still open when $|\Sigma|$ becomes super constant. One open problem that has received significant attention in the literature is whether it is hard to approximate 2-CSPs to within a polynomial factor of both $|\Sigma|$ and $|V|$ (i.e. $(|\Sigma||V|)^{\Omega(1)}$ factor). As a special case of what is now referred to as the Sliding Scale Conjecture, Bellare \etal~\cite{BGLR93} suggested that the answer to this question might be positive. Alas, despite many efforts by researchers to resolve this conjecture (e.g. \cite{RazS97, ArSu, DFKRS11,DHK15,Mos17}), it still remains open to this day.

In this work, we separate between $\card V$ and $\card \Sigma$ and ask a closely related but weaker question: is it hard to approximate 2-CSPs to within a polynomial factor of $|V|$ (but while $\card\Sigma$ may be super-polynomial in $\card V$)? Assuming the exponential time hypothesis (ETH), we answer this question positively. Specifically, we show that, unless ETH fails, no polynomial time algorithm can approximate 2-CSPs to within a factor of $|V|^{1 - 1/(\log |V|)^{1/2 - \rho}}$ for every $\rho > 0$. Note that our ratio is not only polynomial but also almost linear. This is almost optimal since a trivial algorithm yields an $O(|V|)$-approximation for 2-CSPs.

Thanks to a known reduction~\cite{DK99,CFM17} from 2-CSPs to the Directed Steiner Network (DSN) problem (aka Directed Steiner Forest), our result implies an inapproximability result for the latter with polynomial ratio in terms of the number of demand pairs. Specifically, assuming ETH, no polynomial time algorithm can approximate DSN to within a factor of $k^{1/4 - o(1)}$ where $k$ is the number of demand pairs. The ratio is roughly the square root of the best known approximation ratio achieved by polynomial time algorithms due to Chekuri \etal~\cite{CEGS11} and Feldman \etal~\cite{FKN12}, which yield $O(k^{1/2 + \varepsilon})$-approximation for every constant $\varepsilon > 0$.

Additionally, if we assume a stronger assumption that there exists $\varepsilon > 0$ such that no subexponential time algorithm can distinguish a satisfiable 3-CNF formula from one which is not even $(1 - \varepsilon)$-satisfiable (aka Gap-ETH), then, for 2-CSPs, our reduction not only rules out polynomial time (i.e. $(|V||\Sigma|)^{O(1)}$ time) algorithms, but also fixed parameter tractable (FPT) algorithms parameterized by the number of variables $|V|$. These are algorithms with running time $t(|V|) \cdot (|V||\Sigma|)^{O(1)}$ where $t$ can be any function. Similar improvements also apply for DSN parameterized by the number of demand pairs $k$.

\end{abstract}

\newpage
\setcounter{page}{1}

\section{Introduction}

We study the 2-ary constraint satisfaction problems (2-CSPs), which can be stated as follows: given a constraint graph $G = (V, E)$, an alphabet set $\Sigma$ and, for each edge $\{u, v\} \in E$, a constraint $C_{uv} \subseteq \Sigma \times \Sigma$, the goal is to find an assignment $\sigma: V \to \Sigma$ that satisfies as many constraints as possible, where a constraint $C_{uv}$ is said to be satisfied by $\sigma$ if $(\sigma(u), \sigma(v)) \in C_{uv}$. Throughout the paper, we use $k$ to denote the number of variables $|V|$, $n$ to denote the the alphabet size $|\Sigma|$, and $N$ to denote the instance size $nk$.

Constraint satisfaction problems and their inapproximability have been studied extensively since the proof of the PCP theorem in the early 90's \cite{AS,ALMSS}. Most of the effort has been directed towards understanding the approximability of CSPs with constant arity and constant alphabet size, leading to a reasonable if yet incomplete understanding of the landscape \cite{Hastad, Khot-ugc,KKMO, Rag08, AM, Chan}. When the alphabet size grows, the sliding scale conjecture of  \cite{BGLR93} predicts that the hardness of approximation ratio will grow as well, and be at least polynomial\footnote{Througout the paper, we use \emph{polynomial in $x$} (or $\poly(x)$) to refer to $x^c$ for some real number $c > 0$.} in the alphabet size $n$. This has been confirmed for values of $n$ up to $2^{(\log N)^{1-\delta}}$, see \cite{RazS97, ArSu, DFKRS11}. Proving the same for $n$ that is polynomial in $N$ is the so-called polynomial sliding scale conjecture and is still quite open. Before we proceed, let us note that the aforementioned results of \cite{RazS97, ArSu, DFKRS11} work only for arity strictly larger than two and, hence, do not imply inapproximability for 2-CSPs. We will discuss the special case of 2-CSPs in details below.

The polynomial sliding scale conjecture has been approached from different angles. In \cite{DHK15} the authors try to find the smallest arity and alphabet size such that the hardness factor is polynomial in $n$, and in \cite{D16} the conjecture is shown to follow (in some weaker sense) from the Gap-ETH hypothesis, which we discuss in more details later. In this work we focus on yet another angle, which is to separate $n$ and $k$ and ask whether it is hard to approximate constant arity CSPs to within a factor that is polynomial in $k$ (but possibly not polynomial in $n$). Observe here that obtaining \NP-hardness of $\poly(k)$ factor is likely to be as hard as obtaining one with $\poly(N)$; this is because CSPs can be solved exactly in time $n^{O(k)}$, which means that, unless $\NP$ is contained in subexponential time (i.e. $\NP \nsubseteq \bigcap_{\varepsilon > 0} \DTIME(2^{n^{\varepsilon}})$), $\NP$-hard instances of CSPs must have $k = \poly(N)$.

This motivates us to look for hardness of approximation from assumptions stronger than $\P \ne \NP$. Specifically, our result will be based on the Exponential Time Hypothesis (ETH), which states that no subexponential time algorithm can solve 3-SAT (see Conjecture~\ref{conj:eth}). We show that, unless ETH fails, no polynomial time algorithm can approximate 2-CSPs to within an almost linear ratio in $k$, as stated below. This is almost optimal since there is a straightforward $(k/2)$-approximation for any 2-CSP, by simply satisfying all constraints that touch the variable with highest degree.

\begin{theorem}[Main Theorem] \label{thm:eth-hardness}
Assuming ETH, for any constant $\rho > 0$, no algorithm can, given a 2-CSP instance $\Gamma$ with alphabet size $n$ and $k$ variables such that the constraint graph is the complete graph on the $k$ variables, distinguish between the following two cases in polynomial time:
\begin{itemize}
\item (Completeness) $\val(\Gamma) = 1$, and,
\item (Soundness) $\val(\Gamma) < 2^{(\log k)^{1/2 + \rho}}/k$.
\end{itemize}
Here $\val(\Gamma)$ denotes the maximum fraction of edges satisfied by any assignment.
\end{theorem}

To paint a full picture of how our result stands in comparison to previous results, let us state what is know about the approximability of 2-CSPs; due to the vast literature regarding 2-CSPs, we will focus only the regime of large alphabets which is most relevant to our setting. In terms of \NP-hardness, the best known inapproximability ratio is $(\log N)^c$ for every constant $c > 0$; this follows from Moshkovitz-Raz PCP~\cite{MR10} and the Parallel Repetition Theorem for the low soundness regime~\cite{DS14}. Assuming a slightly weaker assumption that $\NP$ is not contained in quasipolynomial time (i.e. $\NP \nsubseteq \bigcup_{c > 0} \DTIME(n^{(\log n)^c})$), 2-CSP is hard to approximate to within a factor of $2^{(\log N)^{1 - \delta}}$ for every constant $\delta > 0$; this can be proved by applying Raz's original Parallel Repetition Theorem~\cite{Raz98} to the PCP Theorem. In \cite{D16}, the author observed that running time for parallel repetition can be reduced by looking at unordered sets instead of ordered tuples. This observation implies that\footnote{In \cite{D16}, only the Gap-ETH-hardness result is stated. However, the ETH-hardness result follows rather easily.}, assuming ETH, no polynomial time $N^{1/(\log\log\log N)^c}$-approximation algorithm exists for 2-CSPs for some constant $c > 0$. Moreover, under Gap-ETH (which will be stated shortly), it was shown that, for every sufficiently small $\varepsilon > 0$, an $N^{\varepsilon}$-approximation algorithm must run in time $N^{\Omega(\exp(1/\varepsilon))}$. Note that, while this latest result comes close to the polynomial sliding scale conjecture, it does not quite resolve the conjecture yet. In particular, even the weak form of the conjecture postulates that there exists $\delta > 0$ for which no polynomial time algorithm can approximate 2-CSPs to within $N^\delta$ factor of the optimum. This statement does not follow from the result of~\cite{D16}. Nevertheless, the Gap-ETH-hardness of~\cite{D16} does imply that, for any $f = o(1)$, no polynomial time algorithm can approximate 2-CSPs to within a factor of $N^{f(N)}$.
%it does not rule out polynomial time $N^\delta$-approximation algorithm for 2-CSPs for any $\delta > 0$; it is still possible that there exists $N^\delta$-approximation algorithm, but the result only suggests that the running time of such algorithm must be large.

In all hardness results mentioned above, the constructions give 2-CSP instances in which the alphabet size $n$ is smaller than the number of variables $k$. In other words, even if we aim for an inapproximability ratio in terms of $k$ instead of $N$, we still get the same ratios as stated above. Thus, our result is the first hardness of approximation for 2-CSPs with $\poly(k)$ factor. Note again that our result rules out any polynomial time algorithm and not just $N^{O(\exp(1/\varepsilon))}$-time algorithm ruled out by~\cite{D16}. Moreover, our ratio is almost linear in $k$ whereas the result of~\cite{D16} only holds for $\varepsilon$ that is sufficiently small depending on the parameters of the Gap-ETH Hypothesis.

%2-CSPs with large alphabets are not only studied from the lower bound perspective. Approximation algorithms for 2-CSPs with large alphabets have also been devised~\cite{Peleg07,CHK11,MM17,CMMV17}. In particular, the best approximation algorithm known, due to Charikar \etal~\cite{CHK11}, yields an $O((nk)^{1/3})$-approximation for the problem in polynomial time. In a special case where the graph satisfies certain expansion properties, it is known that the ratio can be improved to $O((nk)^{1/4 + \varepsilon})$ for every $\varepsilon > 0$~\cite{CMMV17}. We note that, to the best of our knowledge, if one wants the approximation ratio to depends  no polynomial time $o(k)$-approximation is known

%We remark here that our results hold not only for  ~\cite{MM15}

%which is to let the alphabet grow to be even larger than polynomial in the number of variables $k$. In this regime one doesn't expect the hardness factor to be inverse polynomial in $\card\Sigma$ but it can still be polynomial in $k$. The reason is that the number of constant-arity constraints is at most polynomial in $k$ and by satisfying even one constraint (which can always be done in linear time) we already have satisfied a polynomial fraction (in $k$) of all constraints. 

An interesting feature of our reduction is that it produces 2-CSP instances with the alphabet size $n$ that is much larger than $k$. This is reminiscence of the setting of 2-CSPs parameterized by the number of variables $k$. In this setting, the algorithm's running time is allowed to depend not only polynomially on $N$ but also on any function of $k$ (i.e. $g(k) \cdot \poly(N)$ running time for some function $g$); such algorithm is called a \emph{fixed parameter tractable (FPT)} algorithm parameterized by $k$. The question here is whether this added running time can help us approximate the problem beyond the $O(k)$ factor achieved by the straightforward algorithm. We show that, even in this parameterized setting, the trivial algorithm is still essentially optimal (up to lower order terms). This result holds under the Gap Exponential Time Hypothesis (Gap-ETH), a strengthening of ETH which states that, for some $\varepsilon > 0$, even distinguishing between a satisfiable 3-CNF formula and one which is not even $(1 - \varepsilon)$-satisfiable cannot be done in subexponential time (see Conjecture~\ref{conj:gap-eth}), as stated below.

\begin{theorem} \label{thm:gap-eth-hardness}
Assuming Gap-ETH, for any constant $\rho > 0$ and any function $g$, no algorithm can, given a 2-CSP instance $\Gamma$ with alphabet size $n$ and $k$ variables such that the constraint graph is the complete graph on the $k$ variables, distinguish between the following two cases in $g(k) \cdot (nk)^{O(1)}$ time:
\begin{itemize}
\item (Completeness) $\val(\Gamma) = 1$, and,
\item (Soundness) $\val(\Gamma) < 2^{(\log k)^{1/2 + \rho}}/k$.
\end{itemize}
\end{theorem}

To the best of our knowledge, the only previous inapproximability result for parameterized 2-CSPs is from~\cite{CFM17}. There the authors showed that, assuming Gap-ETH, no $k^{o(1)}$-approximation  $g(k) \cdot (nk)^{O(1)}$-time algorithm exists; this is shown via a simple reduction from parameterized inapproximbability of Densest-$k$ Subgraph from~\cite{CCKLMNT17} (which is in turn based on a construction from~\cite{M17}). Our result is a direct improvement over this result.

We end our discussion on 2-CSPs by noting that several approximation algorithms have also been devised for 2-CSPs with large alphabets~\cite{Peleg07,CHK11,KKT16,MM17,CMMV17}. In particular, while our results suggest that the trivial algorithm achieves an essentially optimal ratio in terms of $k$, non-trivial approximation is possible when we measure the ratio in terms of $N$ instead of $k$: specifically, a polynomial time $O(N^{1/3})$-approximation algorithm is known~\cite{CHK11}.

\paragraph{Direct Steiner Network.} As a corollary of our hardness of approximation results for 2-CSPs, we obtain an inapproximability result for Directed Steiner Network with polynomial ratio in terms of the number of demand pairs. In the Directed Steiner Network (DSN) problem (sometimes referred to as the Directed Steiner Forest problem~\cite{FKN12,CDKL17}), we are given an edge-weighed directed graph $G$ and a set $\cD$ of $k$ demand pairs $(s_1, t_1), \dots, (s_k, t_k) \in V \times V$ and the goal is to find a subgraph $H$ of $G$ with minimum weight such that there is a path in $H$ from $s_i$ to $t_i$ for every $i \in [k]$. DSN was first studied in the approximation algorithms context by Charikar \etal~\cite{CCCDGGL99} who gave a polynomial time $\tO(k^{2/3})$-approximation algorithm for the problem. This ratio was later improved to $O(k^{1/2 + \varepsilon})$ for every $\varepsilon > 0$ by Chekuri \etal~\cite{CEGS11}. Later, a different approximation algorithm with similar approximation ratio was proposed by Feldman \etal~\cite{FKN12}.

Algorithms with approximation ratios in terms of the number of vertices $n$ have also been devised~\cite{FKN12,BBMRY13,CDKL17,AB17}. In this case, the best known algorithm is that of Berman \etal~\cite{BBMRY13}, which yields an $O(n^{2/3 + \varepsilon})$-approximation for every constant $\varepsilon > 0$ in polynomial time. Moreover, when the graph is unweighted (i.e. each edge costs the same), Abboud and Bodwin recently gave an improved $O(n^{0.5778})$-approximation algorithm for the problem~\cite{AB17}.

On the hardness side, there exists a known reduction from 2-CSP to DSN that preserves approximation ratio to within polynomial factor\footnote{That is, for any non-decreasing function $\rho$ , if DSN admits $\rho(nk)$-approximation in polynomial time, then 2-CSP also admits $\rho(nk)^c$-approximation polynomial time for some absolute constant $c$.}~\cite{DK99}. Hence, known hardness of approximation of 2-CSPs translate immediately to that of DSN: it is \NP-hard to approximate to within any polylogarithmic ratio~\cite{MR10,DS14}, it is hard to approximate to within $2^{\log^{1 - \varepsilon} n}$ factor for every $\varepsilon > 0$ unless $\NP \subseteq \QP$~\cite{Raz98}, and it is Gap-ETH-hard to approximate to within $n^{o(1)}$ factor~\cite{D16}. Note that, since $k$ is always bounded above by $n^2$, all these hardness results also hold when $n$ is replaced by $k$ in the ratios. Recently, this reduction was also used by Chitnis \etal~\cite{CFM17} to rule out $k^{o(1)}$-FPT-approximation algorithm for DSN parameterized by $k$ assuming Gap-ETH. Alas, none of these hardness results achieve ratios that are polynomial in either $n$ or $k$ and it remains open whether DSN is hard to approximate to within a factor that is polynomial in $n$ or in $k$.

By plugging our hardness result for 2-CSPs into the reduction, we immediately get ETH-hardness and Gap-ETH-hardness of approximating DSN to within a factor of $k^{1/4 - o(1)}$ as stated below.

\begin{corollary} \label{cor:eth-hardness-dsn}
Assuming ETH, for any constant $\rho' > 0$, there is no polynomial time $\frac{k^{1/4}}{2^{(\log k)^{1/2 + \rho'}}}$-approximation algorithm for DSN.
\end{corollary}

\begin{corollary} \label{cor:gap-eth-hardness-dsn}
Assuming Gap-ETH, for any constant $\rho' > 0$ and any function $g$, there is no $g(k) \cdot (nk)^{O(1)}$-time $\frac{k^{1/4}}{2^{(\log k)^{1/2 + \rho'}}}$-approximation algorithm for DSN.
\end{corollary}

In other words, if one wants a polynomial time approximation algorithm with ratio depending only on $k$ and not on $n$, then the algorithms of Chekuri \etal~\cite{CEGS11} and Feldman \etal~\cite{FKN12} are roughly within a square of the optimal algorithm.
To the best of our knowledge, these are the first inapproximability results of DSN whose ratios are polynomial in terms of $k$. Again, Corollary~\ref{cor:gap-eth-hardness-dsn} is a direct improvement over the FPT inapproximability result from~\cite{CFM17} which, under the same assumption, rules out only $k^{o(1)}$-factor FPT-approximation algorithm. %Let us also note here that there exists an algorithm that solve DSN exactly in $n^{O(k)}$ time~\cite{FR06}. Hence, similar to 2-CSPs, unless $\NP$ is contained in subexponential time (i.e. $\cap_{\varepsilon > 0} \DTIME(2^{n^\varepsilon})$), \NP-hardness of approximating DSN to within polynomial ratio in $k$ is equivalent to \NP-hardness of approximating DSN to within polynomial ratio in $n$. In other words, proving $\poly(k)$-factor \NP-hardness of approximating DSN is likely to be as hard as $\poly(n)$-factor \NP-hardness of approximation.

\subsubsection*{Agreement tests}
Our main result is proved through an agreement testing argument. In agreement testing there is a universe $\cU$, a collection of subsets $S_1,\ldots,S_k\subseteq \cU$, and for each subset $S_i$ we are given a local function $f_{S_i} : S_i \to \{0,1\}$. A pair of subsets are said to {\em agree} if their local functions agree on every element in the intersection. The goal is, given a non-negligible fraction of agreeing pairs, to deduce the existence of a global function $g:\cU\to\{0,1\}$ that (approximately) coincides with many of the local functions. For a more complete description see \cite{DK17}.

Agreement tests capture a natural local to global statement and are present in essentially all PCPs, for example they appear explicitly in the line vs. line and plane vs. plane low degree tests \cite{RuSu96,ArSu,RazS97}. Our reduction is based on a combinatorial agreement test, where the universe is $[n]$ and the subsets $S_1,\ldots,S_k$ have $\Omega(n)$ elements each and are ``in general position'', namely they behave like subsets chosen independently at random. A convenient feature about this setting is that every pair of subsets intersect.

Since we are aiming for a large gap, the agreement test must work (i.e., yield a global function) with a very small fraction of agreeing pairs, which in our case is close to $1/k$.

In this small agreement regime the idea, as pioneered in the work of Raz-Safra~\cite{RazS97}, is to zero in on a sub-collection of subsets that is (almost) perfectly consistent. From this sub-collection it is easy to recover a global function and show that it coincides almost perfectly with the local functions in the sub-collection. A major difference between our combinatorial setting and the algebraic setting of Raz-Safra is the lack of ``distance'' in our case: we can not assume that two distinct local functions differ on many points (in contrast, this is a key feature of low degree polynomials). We overcome this by considering different ``strengths'' of agreement, depending on the fraction of points on which the two subsets agree. This notion too is present in several previous works on combinatorial agreement tests \cite{ImpagliazzoKW12,DN17}.
\inote{would be nice to formulate and prove an agreement theorem in the full version}

\paragraph{Hardness of Approximation through Subexponential Time Reductions.} Our result is one of the many results in recent years that show hardness of approximation via subexponential time reductions~\cite{AIM14,BKW15,Rub16-focs,DFS16,D16,BKRW17,MR17,M17,Rub-itcs,Rub16-focs,Rub17,AR17,CCKLMNT17,CLM17,Rub18,BGKM18}. These results are often based on the Exponential Time Hypothesis (ETH) and its variants. Proposed by Impagliazzo and Paturi~\cite{IP01}, ETH can be formally stated as follows:

\begin{conjecture}[Exponential Time Hypothesis (ETH)~\cite{IP01}] \label{conj:eth}
There exists a constant $\delta > 0$ such that no algorithm can decide whether any given 3-CNF formula is satisfiable in time $O(2^{\delta m})$ where $m$ denotes the number of clauses\footnote{The original conjecture states the lower bound as exponential in terms of the number of variables not clauses. However, thanks to the sparsification lemma~\cite{IPZ01}, it is by now known that the two versions are equivalent.}.
\end{conjecture}

A crucial ingredient in most, but not all\footnote{The exceptions are \cite{Rub16-focs,AR17,Rub18,CLM17,Chen18} in which gaps are not created via the PCP Theorem.}, reductions in this line of work is a nearly-linear size PCP Theorem. For the purpose of our work, the PCP Theorem can be viewed as a polynomial time transformation of a 3-SAT instance $\tPhi$ to another 3-SAT instance $\Phi$ that creates a gap between the YES and NO cases. Specifically, if $\tPhi$ is satisfiable, $\Phi$ remains satisfiable. On the other hand, if $\tPhi$ is unsatisfiable, then $\Phi$ is not only unsatisfiable but it is also not even $(1 - \varepsilon)$-satisfiable for some constant $\varepsilon > 0$ (i.e. no assignment satisfies $(1 - \varepsilon)$ fraction of clauses). The ``nearly-linear size'' part refers to the size of the new instance $\Phi$ compared to that of $\tPhi$. Currently, the best known dependency in this form of the PCP Theorem between the two sizes is quasi-linear (i.e. with a polylogarithmic blow-up), as stated below.

\begin{theorem}[Quasi-Linear Size PCP~\cite{BS08,D07}] \label{thm:d-pcp}
For some constants $\varepsilon, \Delta, c > 0$, there is a polynomial time algorithm that, given any 3-CNF formula $\tPhi$ with $m$ clauses, produces another 3-CNF formula $\Phi$ with $O(m \log^c m)$ clauses such that
\begin{itemize}
\item (Completeness) if $\val(\tPhi) = 1$, then $\val(\Phi) = 1$, and,
\item (Soundness) if $\val(\tPhi) < 1$, then $\val(\Phi) < 1 - \varepsilon$, and,
\item (Bounded Degree) each variable in $\Phi$ appears in at most $\Delta$ clauses.
\end{itemize}
\end{theorem}

The aforementioned ETH-hardness of approximation proofs typically proceed in two steps. First, the PCP Theorem is invoked to reduce a 3-SAT instance $\tPhi$ of size $m$ to an instance of the gap version of 3-SAT $\Phi$ of size $m' = O(m \log^c m)$. Second, the gap version of 3-SAT is reduced in subexponential time to the problem at hand. As long as the reduction takes time $2^{o(m'/\log^c m')} = 2^{o(m)}$, we can obtain hardness of approximation result for the latter problem. This is in contrast to proving \NP-hardness of approximation for which a polynomial time reduction is required.

Another related but stronger version of ETH that we will also employ is the Gap Exponential Time Hypothesis (Gap-ETH), which states that even the gap version of 3-SAT cannot be solved in subexponential time:

\begin{conjecture}[Gap Exponential Time Hypothesis (Gap-ETH)~\cite{D16,MR16}] \label{conj:gap-eth}
There exist constants $\delta, \varepsilon, \Delta > 0$ such that no algorithm can, given any 3-CNF formula $\Phi$ such that each of its variable appears in at most $\Delta$ clauses\footnote{This bounded degree assumption can be assumed without loss of generality; see~\cite{MR16} for more details.}, distinguish between the following two cases\footnote{Note that when $\Phi$ satisfies neither case (i.e. $1 - \varepsilon \leqs \val(\Phi) < 1$), the algorithm is allowed to output anything.} in time $O(2^{\delta m})$ time where $m$ denotes the number of clauses:
\begin{itemize}
\item (Completeness) $\val(\Phi) = 1$.
\item (Soundness) $\val(\Phi) < 1 - \varepsilon$.
\end{itemize}
\end{conjecture}

By starting with Gap-ETH instead of ETH, there is no need to apply the PCP Theorem and hence a polylogarithmic loss in the size of the 3-SAT instance does not occur. As demonstrated in previous works, this allows one to improve the ratio in hardness of approximation results~\cite{D16,MR16,M17} and, more importantly, it can be used to prove inapproximability results for some parameterized problems~\cite{BEKP15,CCKLMNT17,CFM17}\footnote{While~\cite{BEKP15} states that the assumption is the existence of a linear-size PCP, Gap-ETH clearly suffices there.}, which are not known to be hard to approximate under ETH. Specifically, for many parameterized problems, the reduction from the gap version of 3-SAT to the problem has size $2^{m'/f(k)}$ for some function $f$ that grows to infinity with $k$ (i.e. $f = \omega(1)$), where $m'$ is the number of clauses in the 3-CNF formula and $k$ is the parameter of the problem. For simplicity, let us focus on the case where $f(k) = k$. If one wishes to derive a meaningful result starting form ETH, $2^{m'/k}$ must be subexponential in terms of $m$, the number of clauses in the original (no-gap) 3-CNF formula. This means that the term $k$ must dominate the $\log^c m$ factor blow-up from the PCP Theorem. However, since FPT algorithms are allowed to have running time of the form $g(k)$ for any function $g$, we can pick $g$ to be $2^{2^k}$. In this case, the algorithm runs in superexponential time in terms of $m$ and we cannot deduce anything regarding the algorithm. On the other hand, if we start from Gap-ETH, we can pick $k$ to be a large constant independent of $m$, which indeed yields hardness of the form claimed in Theorem~\ref{thm:gap-eth-hardness} and Corollary~\ref{cor:gap-eth-hardness-dsn}.

Finally, we remark that Gap-ETH would follow from ETH if a linear-size (constant-query) PCP exists. While constructing short PCPs has long been an active area of research~\cite{BGHSV06,BS08,D07,MR10,BKKMS16}, no linear-size PCP is yet known. On the other hand, there are some supporting evidences for the hypothesis. For instance, it is known that the natural relaxation of 3-SAT in the Sum-of-Squares hierarchy cannot refute Gap-ETH~\cite{Gri01,Sch08}. Moreover, Applebaum recently showed that the hypothesis follows from certain cryptographic assumptions~\cite{App17}. For a more in-depth discussion on Gap-ETH, please refer to~\cite{D16}.

\paragraph{Organization of the Paper.} In the next section, we describe our reduction and give an overview of the proof. Then, in Section~\ref{sec:prelim}, we define additional notions and state some preliminaries. We proceed to provide the full proof of our main agreement theorem in Section~\ref{sec:soundness}. Using this agreement theorem, we deduce the soundness of our reduction in Section~\ref{sec:soundness-csp}. We then plug in the parameters and prove the inapproximability results for 2-CSPs in Section~\ref{sec:2csp-hardness}. In Section~\ref{sec:dsn}, we show how the hardness of approximation result for 2-CSPs imply inapproximability for DSN as well. Finally, we conclude our work with some discussions and open questions in Section~\ref{sec:conclusion}.

\section{Proof Overview} \label{sec:overview}

Like other (Gap-)ETH-hardness of approximation results, our proof is based on a subexponential time reduction from the gap version of 3-SAT to our problem of interest, 2-CSPs. 
%Combining a nearly-linear size PCP with such a reduction gives us the desired inapproximability of 2-CSP. 
Before we describe our reduction, let us define more notations for 2-CSPs and 3-SAT, to facilitate our explanation.

{\bf 2-CSPs.} For notational convenience, we will modify the definition of 2-CSPs slightly so that each variable is allowed to have different alphabets; this definition is clearly equivalent to the more common definition used above. Specifically, an instance $\Gamma$ of 2-CSP now consists of (1) a constraint graph $G = (V, E)$, (2) for each vertex (or variable) $v \in V$, an alphabet set $\Sigma_v$, and, (3) for each edge $\{u, v\} \in E$, a constraint $C_{uv} \subseteq \Sigma_u \times \Sigma_v$. Additionally, to avoid confusion with 3-SAT, we refrain from using the word \emph{assignment} for 2-CSPs and instead use \emph{labeling}, i.e., a labeling of $\Gamma$ is a tuple $\sigma = (\sigma_v)_{v \in V}$ such that $\sigma_v \in \Sigma_v$ for all $v \in V$. An edge $\{u, v\} \in E$ is said to be \emph{satisfied} by a labeling $\sigma$ if $(\sigma_u, \sigma_v) \in \Sigma_u \times \Sigma_v$. The value of a labeling $\sigma$, denoted by $\val(\sigma)$, is defined as the fraction of edges that it satisfies, i.e., $|\{\{u, v\} \in E \mid (\sigma_u, \sigma_v) \in C_{uv}\}| / |E|$. The goal of 2-CSPs is to find $\sigma$ with maximum value; we denote the such optimal value by $\val(\Gamma)$, i.e., $\val(\Gamma) = \max_{\sigma} \val(\sigma)$.

{\bf 3-SAT.} An instance $\Phi$ of 3-SAT consists of a variable set $\cX$ and a clause set $\cC$ where each clause is a disjunction of at most three literals. For any assignment $\psi: \cX \to \{0, 1\}$, $\val(\psi)$ denotes the fraction of clauses satisfied by $\psi$. The goal is to find an assignment $\psi$ that satisfies as many clauses as possible; let $\val(\Phi) = \max_{\psi} \val(\psi)$ denote the fraction of clauses satisfied by such assignment. For each $C \in \cC$, we use $\var(C)$ to denote the set of variables whose literals appear in $C$. We extend this notation naturally to sets of clauses, i.e., for every $T \subseteq \cC$, $\var(T) = \bigcup_{C \in T} \var(C)$.

\subsection*{Our Construction}
Before we state our reduction, let us again reiterate the objective of our reduction. Roughly speaking, given a 3-SAT stance $\Phi = (\cX, \cC)$, we would like to produce a 2-CSP instance $\Gamma_{\Phi}$ such that
\begin{itemize}
\item (Completeness) If $\val(\Phi) = 1$, then $\val(\Gamma_\Phi) = 1$,
\item (Soundness) If $\val(\Phi) < 1 - \varepsilon$, then $\val(\Gamma_\Phi) < k^{o(1)}/k$ where $k$ is number of variables of $\Gamma_\Phi$,
\item (Reduction Time) The time it takes to produce $\Gamma_{\Phi}$ should be $2^{o(m)}$ where $m = |\cC|$,
\end{itemize}
where $\varepsilon > 0$ is some absolute constant.

Observe that, when plugging a reduction with these properties to Gap-ETH, we directly arrive at the claimed $k^{1 - o(1)}$ inapproximability for 2-CSPs. However, for ETH, since we start with a decision version of 3-SAT without any gap, we have to first invoke the PCP theorem to produce an instance of the gap version of 3-SAT before we can apply our reduction. Since the shortest known PCP has a polylogarithmic blow-up in the size (see Theorem~\ref{thm:d-pcp}), the running time lower bound for gap 3-SAT will not be exponential anymore, rather it will be of the form $2^{\Omega(m/\polylog m)}$ instead. Hence, our reduction will need to produce $\Gamma_\Phi$ in $2^{o(m/\polylog m)}$ time. As we shall see later in Section~\ref{sec:2csp-hardness}, this will also be possible with appropriate settings of parameters.

With the desired properties in place, we now move on to state our reduction. In addition to a 3-CNF formula $\Phi$, the reduction also takes in a collection $\cT$ of subsets of clauses of $\Phi$. For now, the readers should think of the subsets in $\cT$ as random subsets of $\cC$ where each element is included in each subset independently at random with probability $\alpha$, which will be specified later. As we will see below, we only need two simple properties that the subsets in $\cT$ are ``well-behaved'' enough and we will later give a deterministic construction of such well-behaved subsets. With this in mind, our reduction can be formally described as follows.

\begin{definition}[The Reduction] \label{def:construction}
Given a 3-CNF formula $\Phi = (\cX, \cC)$ and a collection $\cT$ of subsets of $\cC$, we define a 2-CSP instance $\Gamma_{\Phi, \cT} = (G = (V, E), \Sigma, \{C_{uv}\}_{\{u, v\} \in E})$ as follows:
\begin{itemize}
\item The graph $G$ is the complete graph where the vertex set is $\cT$, i.e., $V = \cT$ and $E = \binom{\cT}{2}$.
\item For each $T \in \cT$, the alphabet set $\Sigma_{\cT}$ is the set of all partial assignments to $\var(T)$ that satisfies every clause in $T$, i.e., $\Sigma_T = \{\psi_T: \var(T) \rightarrow \{0, 1\} \mid \forall C \in T, \psi_T \text{ satisfies } C\}$.
\item For every $T_1 \ne T_2 \in \cT$, $(\psi_{T_1}, \psi_{T_2})$ is included in $C_{T_1T_2}$ if and only if they are consistent, i.e., $C_{T_1T_2} = \{(\psi_{T_1}, \psi_{T_2}) \in \Sigma_{T_1} \times \Sigma_{T_2} \mid \forall x \in \var(T_1) \cap \var(T_2), \psi_{T_1}(x) = \psi_{T_2}(x)\}$.
\end{itemize}
\end{definition}

Let us now examine the properties of the reduction. The number of vertices in $\Gamma_{\Phi, \cT}$ is $k = |\cT|$. For the purpose of the proof overview, $\alpha$ should be thought of as $1/\polylog m$ whereas $k$ should be thought of as much larger than $1/\alpha$ (e.g. $k = \exp(1/\alpha))$. For such value of $k$, all random sets in $\cT$ will have size $O(\alpha m)$ w.h.p., meaning that the reduction time is $2^{m/\polylog m}$ as desired.

Moreover, when $\Phi$ is satisfiable, it is not hard to see that $\val(\Gamma_{\Phi, \cT}) = 1$; more specifically, if $\psi: \cX \to \{0, 1\}$ is the assignment that satisfies every clause of $\Phi$, then we can label each vertex $T \in \cT$ of $\Gamma_{\Phi, \cT}$ by $\psi|_{\var(T)}$, the restriction of $\psi$ on $\var(T)$. Since $\psi$ satisfies all the clauses, $\psi|_{\var(T)}$ satisfies all clauses in $T$, meaning that this is a valid labeling. Moreover, since these are restrictions of the same global assignment $\psi$, they are all consistent and every edge is satisfied.

Hence, we are only left to show that, if $\val(\Phi) < 1 - \varepsilon$, then $\val(\Gamma_{\Phi, \cT}) < k^{o(1)}/k$; this is indeed our main technical contribution. We will show this by contrapositive: assuming that $\val(\Gamma_{\Phi, \cT}) \geqs k^{o(1)}/k$, we will ``decode'' back an assignment to $\Phi$ that satisfies $1 - \varepsilon$ fraction of clauses.

\subsection{Soundness Analysis as an Agreement Theorem}

Our task at hand can be viewed as agreement testing. Informally, in agreement testing, the input is a collection $\{f_S\}_{S \in \cS}$  of local functions $f_S: S \to \{0, 1\}$ where $\cS$ is a collection of subsets of some universe $\cU$ such that, for many pairs $S_1$ and $S_2$, $f_{S_1}$ and $f_{S_2}$ agree, i.e., $f_{S_1}(x) = f_{S_2}(x)$ for all $x \in S_1 \cap S_2$. An agreement theorem says that there must be a global function $g: \cU \to \{0, 1\}$ that coincides (exactly or approximately) with many of the local functions, and thus explains the pairwise ``local'' agreements. In our case, a labeling $\sigma = \{\sigma_T\}_{T \in \cT}$ with high value is exactly a collection of functions $\sigma_T: \var(T) \to \{0, 1\}$ such that, for many pairs of $T_1$ and $T_2$, $\sigma_{T_1}$ and $\sigma_{T_2}$ agrees. The heart of our soundness proof is an agreement theorem that recovers a global function $\psi: \cX \to \{0, 1\}$ that approximately coincides with many of the local functions $\sigma_T$'s and thus satisfies $1 - \varepsilon$ fraction of clauses of $\Phi$. To discuss the agreement theorem in more details, let us define several additional notations, starting with those for (approximate) agreements of a pair of functions:

\begin{definition}
For any universe $\cU$, let $f_{S_1}: S_1 \to \{0, 1\}$ and $f_{S_2}: S_2 \to \{0, 1\}$ be any two functions whose domains $S_1, S_2$ are subsets of $\cU$. We use the following notations for (dis)agreements of these two functions:
\begin{itemize}
\item Let $\disa(f_{S_1}, f_{S_2})$ denote the number of $x \in S_1 \cap S_2$ that $f_{S_1}$ and $f_{S_2}$ disagree on, i.e., $\disa(f_{S_1}, f_{S_2}) = |\{x \in S_1 \cap S_2 \mid f_{S_1}(x) \ne f_{S_2}(x)\}|$.
\item For any $\zeta \geqs 0$, we say that $f_{S_1}$ and $f_{S_2}$ are \emph{$\zeta$-consistent} if $\disa(f_{S_1}, f_{S_2}) \leqs \zeta|\cU|$, and we say that the two functions are \emph{$\zeta$-inconsistent} otherwise. For $\zeta = 0$, we sometimes drop 0 and refer to these simply as \emph{consistent} and \emph{inconsistent} (instead of 0-consistent and 0-inconsistent).
\item We use $f_{S_1} \overset{\zeta}{\approx} f_{S_2}$ and $f_{S_1} \overset{\zeta}{\napprox} f_{S_2}$ as shorthands for $\zeta$-consistency and $\zeta$-inconsistency respectively. Again, for $\zeta = 0$, we may drop 0 from the notations and simply use $f_{S_1} \approx f_{S_2}$ and $f_{S_1} \napprox f_{S_2}$.
\end{itemize}
\end{definition}

Next, we define the notion of agreement probability for any collection of functions:

\begin{definition}
For any $\zeta \geqs 0$ and any collection $\cF = \{f_S\}_{S \in \cS}$ of functions, the $\zeta$-agreement probability, denoted by $\agr_{\zeta}(\cF)$ is the probability that $f_{S}$ is $\zeta$-consistent with $f_{S'}$ where $S$ and $S'$ are chosen independently uniformly at random from $\cS$, i.e., $\agr_\zeta(\cF) = \Pr_{S, S' \in \cS}[f_S \overset{\zeta}{\approx} f_{S'}]$. When $\zeta = 0$, we will drop 0 from the notation and simply use $\agr(\cF)$.
\end{definition}

Our main agreement theorem, which works when each $S \in \cS$ is a large ``random'' subset, says that, if $\agr(\cF)$ is noticeably large, then there exists a global function that is approximately consistent with many of the local functions in $\cF$. This is stated more precisely (but still informally) below.

\begin{theorem}[Informal; See Theorem~\ref{thm:agr-main-formal}] \label{thm:agr-informal}
Let $\cS$ be a collection of $k$ independent random $\alpha n$-element subsets of $[n]$. The following holds with high probability: for any $\beta > 0$ and any collection of functions $\cF = \{f_S\}_{S \in \cS}$ such that $\delta \triangleq \agr(\cF) \geqs k^{o_{\beta,\alpha}(1)}/k$, there exist a function $g: [n] \to \{0, 1\}$ and a subcollection $\cS'$ of size $\delta k^{1 - o_{\beta, \alpha}(1)}$ such that $g \overset{\beta}{\approx} f_{S'}$ for all $S' \in \cS'$.
\end{theorem}

To see that Theorem~\ref{thm:agr-informal} implies our soundness, let us view a labeling $\sigma = \{\sigma_T\}_{T \in \cT}$ as a collection $\cF = \{f_S\}_{S \in \cS}$ where $\cS = \{\var(T) \mid T \in \cT\}$ and $f_{\var(T)}$ is simply $\sigma_{T}$. Now, when $\val(\sigma)$ is large, $\agr(\cF)$ is large as well. Moreover, while the sets $S \in \cS$ are not random subsets of variables but rather variable sets of random subsets of clauses, it turns out that these sets are ``well-behaved'' enough for us to apply Theorem~\ref{thm:agr-informal}. This yields a global function $\psi: \cX \to \{0, 1\}$ that are $\beta$-consistent with many $\sigma_T$'s. Note that, if instead of $\beta$-consistency we had exact consistency, then we would have been done because $\psi$ must satisfy all clauses that appear in any $T$ such that $\psi$ is consistent with $\sigma_T$; since there are many such $T$'s and these are random sets, $\psi$ indeed satisfies almost all clauses. A simple counting argument shows that this remains true even with approximate consistency, provided that most clauses appear in at least a certain fraction of such $T$'s (an assumption which holds for random subsets). Hence, the soundness of our reduction follows from Theorem~\ref{thm:agr-informal}, and we devote the rest of this section to outline an overview of its proof.

{\bf Optimality of the parameters of Theorem~\ref{thm:agr-informal}.} Before we proceed to the overview, we would like to note that the size of the subcollection $\cS'$ in Theorem~\ref{thm:agr-informal} is nearly optimal. This is because, we can partition $\cS$ into $1/\delta$ subcollections $\cS_1, \dots, \cS_{1/\delta}$ each of size $\delta k$ and, for each $i \in [1/\delta]$, randomly select a global function $g_i: [n] \to \{0, 1\}$ and let each $f_S$ be the restriction of $g_i$ to $S$ for each $S \in \cS_i$. In this way, we have $\agr(\cF) \geqs \delta k$ and any global function can be (approximately) consistent with at most $\delta k$ local functions. This means that $\cS'$ can be of size at most $\delta k$ in this case and, up to a $k^{o_{\beta, \alpha}(1)}$ multiplicative factor, Theorem~\ref{thm:agr-informal} yields almost a largest possible $\cS'$.
%Before we proceed to the overview, we would like to note that the parameters in Theorem~\ref{thm:agr-informal} are nearly optimal in two senses. First, the size of the subcollection $\cS'$ is nearly optimal; this is because, we can always partition $\cS$ into $1/\delta$ subcollections $\cS_1, \dots, \cS_{1/\delta}$ each of size $\delta k$ and, for each $i \in [1/\delta]$, randomly select a global function $g_i: [n] \to \{0, 1\}$ and let each $f_S$ be the restriction of $g_i$ to $S$ for each $S \in \cS_i$. In this way, we have $\agr(\cF) \geqs \delta k$ and any global function can be (approximately) consistent with at most $\delta k$ local functions. This means that $\cS'$ can be of size at most $\delta k$ in this case and, up to $O_{\beta, \alpha}(1)$ factor, Theorem~\ref{thm:agr-informal} yields almost largest possible $\cS'$. Second, our theorem works for the agreement probability $\delta$ as low as $k^{o_{\beta, \alpha}(1)}/k$, i.e., when each local function is consistent with only $k^{o_{\beta, \alpha}(1)}$ other local functions on average. This is almost optimal, since (by definition) $\delta$ is always at least $1/k$.

\subsection{A Simplified Proof: $\delta \geqs k^{o(1)}/k^{1/2}$ Regime}

We now sketch the proof of Theorem~\ref{thm:agr-informal}. Before we describe how we can find $g$ when $\delta \geqs k^{o_{\beta, \alpha}(1)}/k$, let us sketch the proof assuming a stronger assumption that $\delta \geqs \Theta_{\alpha, \beta}(1)/k^{1/2}$. Note that this simplified proof already implies a $k^{1/2 - o(1)}$ factor ETH-hardness of approximating 2-CSPs. In the next subsection, we will then proceed to refine the arguments to handle smaller values of $\delta$.

Let us consider the \emph{consistency graph} of $\cF$. This is the graph $G^{\cF}$ whose vertex set is $\cS$ and there is an edge between $S_1$ and $S_2$ if and only if $f_{S_1}$ and $f_{S_2}$ are consistent. Note that the number of edges in $G^{\cF}$ is equal to $\frac{k^2 \delta - k}{2}$, where the subtraction of $k$ comes from the fact that $\delta = \agr(\cF)$ includes the agreement of each set and itself (whereas $G^{\cF}$ does not).

Previous works on agreement testers exploit particular structures of the consistency graph to decode a global function. One such property that is relevant to our proof is the notion of \emph{almost transitivity} defined by Raz and Safra in the analysis of their test~\cite{RazS97}. More specifically, a graph $G = (V, E)$ is said to be $q$-transitive for some $q > 0$ if, for every non-edge $\{u, v\}$ (i.e. $\{u, v\} \in \binom{V}{2} \setminus E$), $u$ and $v$ can share at most $q$ common neighbors\footnote{In~\cite{RazS97}, the transitivity parameter $q$ is used to denote the \emph{fraction} of vertices that are neighbors of both $u$ and $v$ rather than the \emph{number} of such vertices as defined here. However, the latter notion will be more convenient for us.}. Raz and Safra showed that their consistency graph is $(k^{1 - \Omega(1)})$-transitive where $k$ denotes the number of vertices of the graph. They then proved a generic theorem regarding $(k^{1 - \Omega(1)})$-transitive graphs that, for any such graph, its vertex set can be partitioned so that the subgraph induced by each partition is a clique and that the number of edges between different partitions is small. Since a sufficiently large clique corresponds to a global function in their setting, they can then immediately deduce that their result.

Observe that, in our setting, a large clique also corresponds to a global function that is consistent with many local functions. In particular, suppose that there exists $\cS' \subseteq \cS$ of size sufficiently large such that $\cS$ induces a clique in $G^{\cF}$. Since $f_{S'}$'s are perfectly consistent with each other for all $S' \in \cS'$, there is a global function $g: [n] \to \{0, 1\}$ that is consistent with all such $f_{S'}$'s. Hence, if we could show that our consistency graph $G^{\cF}$ is $(k^{1 - \Omega(1)})$-transitive, then we could use the same argument as Raz and Safra's to deduce our desired result. Alas, our graph $G^{\cF}$ does not necessarily satisfy this transitivity property; for instance, consider any two sets $S_1, S_2 \in \cS$ and let $f_{S_1}, f_{S_2}$ be such that they disagree on only one variable, i.e., there is a unique $x \in S_1 \cap S_2$ such that $f_{S_1}(x) \ne f_{S_2}(x)$. It is possible that, for every $S \in \cS$ that does not contain $x$, $f_S$ agrees with both $f_{S_1}$ and $f_{S_2}$; in other words, every such $S$ can be a common neighbor of $S_1$ and $S_2$. Since each variable $x$ appears roughly in only $\Theta(\alpha)$ fraction of the sets, there can be as many as $(1 - \Theta(\alpha))k = (1 - o(1))k$ common neighbors of $S_1$ and $S_2$ even when there is no edge between $S_1$ and $S_2$!

Fortunately for us, a weaker statement holds: if $f_{S_1}$ and $f_{S_2}$ disagree on more than $\zeta n$ variables (instead of just one variable as above), then $S_1$ and $S_2$ have at most $O(\ln(1/\zeta)/\alpha)$ common neighbors in $G^{\cF}$. Here $\zeta$ should be thought of as $\beta^2$ times a small constant which will be specified later. To see why this statement holds, observe that, since every $S \in \cS$ is a random subset that includes each clause $x \in [n]$ with probability $\alpha$, Chernoff bound implies that, for every subcollection $\tcS \subseteq \cS$ of size $\Omega(\ln(1/\zeta)/\alpha)$, $\bigcup_{S \in \tcS} S$ contains all but $O(\zeta)$ fraction of variables. Let $\tcS_{S_1, S_2} \subseteq \cS$ denote the set of common neighbors of $S_1$ and $S_2$. It is easy to see that $S_1$ and $S_2$ can only disagree on variables that do not appear in $\bigcup_{S \in \tcS_{S_1, S_2}} S$. If $\tcS_{S_1, S_2}$ is of size $\Omega(\ln(1/\zeta)/\alpha)$, then $\bigcup_{S \in \tcS_{S_1, S_2}} S$ contains all but $O(\zeta)$ fraction of variables, which means that $S_1$ and $S_2$ disagrees only on $O(\zeta)$ fraction of variables. By selecting the constant appropriately inside $O(\cdot)$, we arrive at the claim statement.

In other words, while the transitive property does not hold for every edge, it holds for the edges $\{S_1, S_2\}$ where $f_{S_1}$ and $f_{S_2}$ are $\zeta$-inconsistent. This motivates us to define a two-level consistency graph, where the edges with $\zeta$-inconsistent are referred to as the \emph{red} edges whereas the original edges in $G^{\cF}$ is now referred to as the \emph{blue} edges. We define this formally below.

\begin{definition}[Red/blue Graph]
A red-blue graph is an undirected graph $G = (V, E = E_r \cup E_b)$ where its edge set $E$ is partitioned into two sets $E_r$, the set of red edges, and $E_b$, the set of blue edges. We use the prefixes ``blue-'' and ``red-'' to refer to the quantities of the graph $(V, E_b)$ and $(V, E_r)$ respectively; for instance, $u$ is said to be a blue-neighbor of $v$ if $\{u, v\} \in E_b$.
\end{definition}

\begin{definition}[Two-Level Consistency Graph] \label{def:two-lev}
Given a collection of functions $\cF = \{f_S\}_{S \in \cS}$ and a real number $0 \leqs \zeta \leqs 1$, the two-level consistency graph $G^{\cF, \zeta} = (V^{\cF, \zeta}, E^{\cF, \zeta}_r \cup E^{\cF, \zeta}_b)$ is a red-blue graph defined as follows.
\begin{itemize}
\item The vertex set $V^{\cF, \zeta}$ is simply $\cS$.
\item The blue edges are the consistent pairs $\{S_1, S_2\}$, i.e., $E_b = \{\{S_1, S_2\} \in \binom{\cS}{2} \mid f_{S_1} \approx f_{S_2}\}$.
\item The red edges are the $\zeta$-inconsistent pairs $\{S_1, S_2\}$, i.e., $E_r = \{\{S_1, S_2\} \in \binom{\cS}{2} \mid f_{S_1} \overset{\zeta}{\napprox} f_{S_2}\}$.
\end{itemize}
\end{definition}

Note that $S_1,S_2$ constitute neither a blue nor a red edge when $0 < \disa(f_{S_1}, f_{S_2}) \leqs \zeta n$.

Now, the transitivity property we argue above can be stated as follows: for every red-edge $\{S_1, S_2\}$ of $G^{\cF, \zeta}$, there are at most $O(\ln(1/\zeta)/\alpha)$ different $S$'s such that both $\{S, S_1\}$ and $\{S, S_2\}$ are blue edges. For brevity, let us call any red-blue graph $G = (V, E_r \cup E_b)$ \emph{$q$-red/blue-transitive} if, for every red edge $\{u, v\} \in E_r$, $u$ and $v$ have at most $q$ common blue-neighbors. We will now argue that in any $q$-red/blue-transitive of average blue-degree $d$, there exists a subset $U \subseteq V$ of size $\Omega(d)$ such that only $O(qk/d^2)$ fraction of pairs of vertices in $U$ are red edges.

Before we prove this, let us state why this is useful for decoding the desired global function $g$. Observe that such a subset $U$ of vertices in the two-level consistency graph translates to a subcollection $\cS' \subseteq \cS$ such that, for all but $O(qk/d^2)$ fraction of pairs of sets $S_1, S_2 \subseteq \cS'$, $\{S_1, S_2\}$ does not form a red edge. Recall from definition of red edges that, for such $S_1, S_2$, $f_{S_1}$ and $f_{S_2}$ disagrees on at most $\zeta n$ variables. In other words, $\cS'$ is similar to a clique in the (not two-level) consistency graph, except that (1) $O(qk/d^2)$ fraction of pairs $\{S_1, S_2\}$ are allowed to disagree on as many variables as they like, and (2) even for the rest of pairs, the guarantee now is that they agree on all but at most $\zeta n$ variables, instead of total agreement as in the previous case of clique. Fortunately, this still suffices to find $g$ that is $O(\sqrt{qk/d^2 + \zeta})$-consistent with $\Omega(d)$ functions.
%Since our actual proof does not extract an assignment from such clique-like subgraph but rather a slightly different (i.e. biclique-like) subgraph (see Lemma~\ref{lem:decoding-biclique}), we will not go into details on how to prove this here but we note that
One way construct such a global function is to simply assign each $g(x)$ according to the majority of $f_S(x)$ for all $S \in \cS'$ such that $x \in S$. (This is formalized in Section~\ref{sec:maj-decode}.) Note that in our case $q = O(\ln(1/\zeta)/\alpha)$ and $d = \Omega(\delta k)$. Hence, if we pick $\zeta \ll \beta^2$ and $\delta \gg (q^{1/2}/\beta)/k^{1/2} = O_{\beta, \alpha}(1)/k^{1/2}$, we indeed get a global function $g$ that is $\beta$-consistent with $\Omega(\delta k)$ local functions.

We now move on to sketch how one can find such an ``almost non-red subgraph''. For simplicity, let us assume that every vertex has the same blue-degree (i.e. $(V, E_b)$ is $d$-regular). Let us count the number of \emph{red-blue-blue triangle} (or \emph{rbb triangle}), which is a 3-tuple $(u, v, w)$ of vertices in $V$ such that $\{u, v\}, \{v, w\}$ are blue edges whereas $\{u, w\}$ is a red edge. An illustration of a rbb triangle can be found in Figure~\ref{fig:triangle}. The red/blue transitivity can be used to bound the number of rbb triangles as follows. For each $(u^*, w^*) \in V^2$, since the graph is $q$-red/blue-transitive there are at most $q$ rbb triangle with $u = u^*$ and $w = w^*$. Hence, in total, there can be at most $qk^2$ rbb triangles. As a result, there exists $v^* \in V$ such that the number of rbb triangles $(u, v, w)$ such that $v = v^*$ is at most $qk$. Let us now consider the set $U = N_b(v^*)$ that consists of all blue-neighbors of $v^*$. There can be at most $qk$ red edges with both endpoints in $N_b(v^*)$ because each such edge corresponds to a rbb triangle with $v = v^*$. From our assumption that every vertex has blue degree $d$, we indeed have that $|U| = d$ and that the fraction of pairs of vertices in $U$ that are linked by red edges is $O(qk/d^2)$ as desired. This completes our overview for the case $\delta \geqs \Theta_{\beta,\alpha}(1) / k^{1/2}$.
\inote{this last part may be shortened or skipped perhaps}

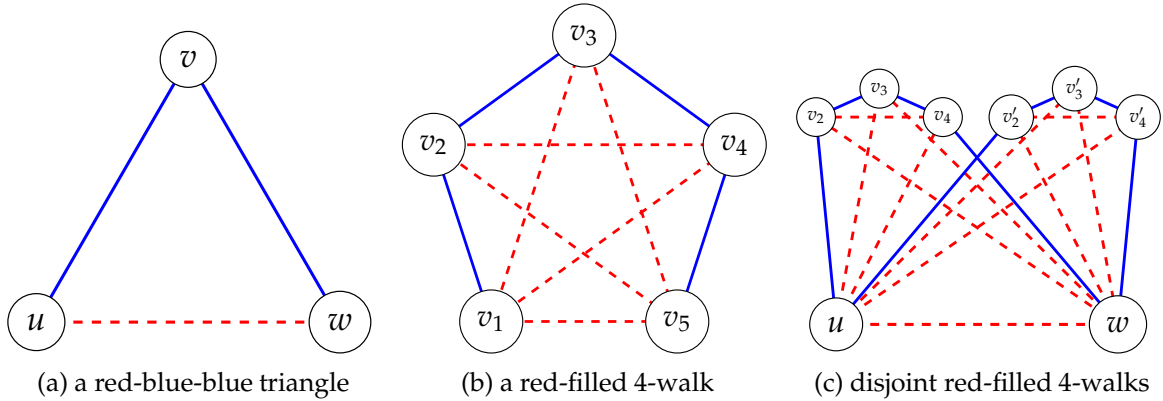
\begin{figure}[h!]
    \centering
    \begin{subfigure}[b]{0.3\textwidth}
        \resizebox{\textwidth}{!}{\begin{tikzpicture}
\draw [line width=1,color=blue] (-1.7320508076, -1.0000000000) -- (0.0000000000, 2.0000000000);
\draw [line width=1,color=red,dashed] (-1.7320508076, -1.0000000000) -- (1.7320508076, -1.0000000000);
\draw [line width=1,color=blue] (0.0000000000, 2.0000000000) -- (1.7320508076, -1.0000000000);
\node[circle,draw=black, fill=white,minimum size=15,scale=1] at (-1.7320508076, -1.0000000000) {$u$};
\node[circle,draw=black, fill=white,minimum size=15,scale=1] at (0.0000000000, 2.0000000000) {$v$};
\node[circle,draw=black, fill=white,minimum size=15,scale=1] at (1.7320508076, -1.0000000000) {$w$};
\end{tikzpicture}}
        \caption{a red-blue-blue triangle}
        \label{fig:triangle}
    \end{subfigure}
    ~
    \begin{subfigure}[b]{0.3\textwidth}
    	\resizebox{\textwidth}{!}{\begin{tikzpicture}
\draw [line width=1,color=blue] (-1.1755705046, -1.6180339887) -- (-1.9021130326, 0.6180339887);
\draw [line width=1,color=red,dashed] (-1.1755705046, -1.6180339887) -- (0.0000000000, 2.0000000000);
\draw [line width=1,color=red,dashed] (-1.1755705046, -1.6180339887) -- (1.9021130326, 0.6180339887);
\draw [line width=1,color=red,dashed] (-1.1755705046, -1.6180339887) -- (1.1755705046, -1.6180339887);
\draw [line width=1,color=blue] (-1.9021130326, 0.6180339887) -- (0.0000000000, 2.0000000000);
\draw [line width=1,color=red,dashed] (-1.9021130326, 0.6180339887) -- (1.9021130326, 0.6180339887);
\draw [line width=1,color=red,dashed] (-1.9021130326, 0.6180339887) -- (1.1755705046, -1.6180339887);
\draw [line width=1,color=blue] (0.0000000000, 2.0000000000) -- (1.9021130326, 0.6180339887);
\draw [line width=1,color=red,dashed] (0.0000000000, 2.0000000000) -- (1.1755705046, -1.6180339887);
\draw [line width=1,color=blue] (1.9021130326, 0.6180339887) -- (1.1755705046, -1.6180339887);
\node[circle,draw=black, fill=white,minimum size=15,scale=1] at (-1.1755705046, -1.6180339887) {$v_1$};
\node[circle,draw=black, fill=white,minimum size=15,scale=1] at (-1.9021130326, 0.6180339887) {$v_2$};
\node[circle,draw=black, fill=white,minimum size=15,scale=1] at (0.0000000000, 2.0000000000) {$v_3$};
\node[circle,draw=black, fill=white,minimum size=15,scale=1] at (1.9021130326, 0.6180339887) {$v_4$};
\node[circle,draw=black, fill=white,minimum size=15,scale=1] at (1.1755705046, -1.6180339887) {$v_5$};
\end{tikzpicture}}
    	\caption{a red-filled $4$-walk}
    	\label{fig:4-walks}
    \end{subfigure}
    ~
    \begin{subfigure}[b]{0.3\textwidth}
    	\resizebox{\textwidth}{!}{\begin{tikzpicture}
\draw [line width=1,color=blue] (-1.7320508076, -1.0000000000) -- (-1.9794228634, 1.5500000000);
\draw [line width=1,color=red,dashed] (-1.7320508076, -1.0000000000) -- (-1.2000000000, 1.9000000000);
\draw [line width=1,color=red,dashed] (-1.7320508076, -1.0000000000) -- (-0.4205771366, 1.5500000000);
\draw [line width=1,color=red,dashed] (-1.7320508076, -1.0000000000) -- (1.7320508076, -1.0000000000);
\draw [line width=1,color=blue] (-1.9794228634, 1.5500000000) -- (-1.2000000000, 1.9000000000);
\draw [line width=1,color=red,dashed] (-1.9794228634, 1.5500000000) -- (-0.4205771366, 1.5500000000);
\draw [line width=1,color=red,dashed] (-1.9794228634, 1.5500000000) -- (1.7320508076, -1.0000000000);
\draw [line width=1,color=blue] (-1.2000000000, 1.9000000000) -- (-0.4205771366, 1.5500000000);
\draw [line width=1,color=red,dashed] (-1.2000000000, 1.9000000000) -- (1.7320508076, -1.0000000000);
\draw [line width=1,color=blue] (-0.4205771366, 1.5500000000) -- (1.7320508076, -1.0000000000);
\draw [line width=1,color=blue] (-1.7320508076, -1.0000000000) -- (0.4205771366, 1.5500000000);
\draw [line width=1,color=red,dashed] (-1.7320508076, -1.0000000000) -- (1.2000000000, 1.9000000000);
\draw [line width=1,color=red,dashed] (-1.7320508076, -1.0000000000) -- (1.9794228634, 1.5500000000);
\draw [line width=1,color=red,dashed] (-1.7320508076, -1.0000000000) -- (1.7320508076, -1.0000000000);
\draw [line width=1,color=blue] (0.4205771366, 1.5500000000) -- (1.2000000000, 1.9000000000);
\draw [line width=1,color=red,dashed] (0.4205771366, 1.5500000000) -- (1.9794228634, 1.5500000000);
\draw [line width=1,color=red,dashed] (0.4205771366, 1.5500000000) -- (1.7320508076, -1.0000000000);
\draw [line width=1,color=blue] (1.2000000000, 1.9000000000) -- (1.9794228634, 1.5500000000);
\draw [line width=1,color=red,dashed] (1.2000000000, 1.9000000000) -- (1.7320508076, -1.0000000000);
\draw [line width=1,color=blue] (1.9794228634, 1.5500000000) -- (1.7320508076, -1.0000000000);
\node[circle,draw=black, fill=white,minimum size=15,scale=1] at (-1.7320508076, -1.0000000000) {$u$};
\node[circle,draw=black, fill=white,minimum size=15,scale=1] at (1.7320508076, -1.0000000000) {$w$};
\node[circle,draw=black, fill=white,minimum size=15,scale=0.6] at (-1.9794228634, 1.5500000000) {$v_2$};
\node[circle,draw=black, fill=white,minimum size=15,scale=0.6] at (-1.2000000000, 1.9000000000) {$v_3$};
\node[circle,draw=black, fill=white,minimum size=15,scale=0.6] at (-0.4205771366, 1.5500000000) {$v_4$};
\node[circle,draw=black, fill=white,minimum size=15,scale=0.6] at (0.4205771366, 1.5500000000) {$v'_2$};
\node[circle,draw=black, fill=white,minimum size=15,scale=0.6] at (1.2000000000, 1.9000000000) {$v'_3$};
\node[circle,draw=black, fill=white,minimum size=15,scale=0.6] at (1.9794228634, 1.5500000000) {$v'_4$};
\end{tikzpicture}}
    	\caption{disjoint red-filled $4$-walks}
    	\label{fig:disjoint-4-walks}
    \end{subfigure}
    \caption{Illustrations of red-filled walks. The red edges are represented by red dashed lines whereas the blue edges are represented by blue solid lines. Figure~\ref{fig:triangle} and Figure~\ref{fig:4-walks} demonstrate a red-filled 2 walk (aka rbb triangle) and a red-filled 4-walk respectively. Figure~\ref{fig:disjoint-4-walks} shows two disjoint red-filled 4-walks.}
\end{figure}

\subsection{Towards $\delta = k^{o(1)}/k$ Regime}

To handle smaller $\delta$, we need to first understand why the approach above fails to work for $\delta \leqs 1/k^{1/2}$. To do so, note that the above proof sketch can be summarized into three main steps:
\begin{enumerate}[(1)]
\item Show that the two-level consistency graph $G^{\cF}$ is $q$-red/blue-transitive for some $q = k^{o(1)}$.
\item Use red/blue transitivity to find a large subgraph of $G^{\cF}$ with few induced red edges. \label{step:finding-almost-clique}
\item Decode a global function from such an ``almost non-red subgraph''.
\end{enumerate}

The reason that we need $\delta \gg 1/k^{1/2}$, or equivalently $d \gg k^{1/2}$, lies in Step~\ref{step:finding-almost-clique}. Although not stated as such earlier, our argument in this step can be described as follows. We consider all length-2 blue-walks, i.e., all $(u, v, w) \in V^3$ such that $\{u, v\}$ and $\{v, w\}$ are both blue edges, and, using the red/blue transitivity of the graph, we argue that, for almost of all these walks, $\{u, w\}$ is not a red edge (i.e. $(u, v, w)$ is not a rbb triangle), which then allows us to find an almost non-red subgraph. For this argument to work, we need the number of length-2 blue-walks to far exceed the number of rbb triangles. The former is $kd^2$ whereas the latter is bounded above by $k^2q$ in $q$-red/blue-transitive graphs. This means that we need $kd^2 \gg k^2q$, which implies that $d \gg k^{1/2}$.

To overcome this limitation, we instead consider all length-$\ell$ blue-walks for $\ell > 2$ and we will define a ``rbb-triangle-like'' structure on these walks. Our goal is again to show that this structure appears rarely in random length-$\ell$ blue-walks and we will then use this to find a subgraph that allows us to decode a good assignment for $\Phi$. Observe that the number of length-$\ell$ blue walks is $kd^{\ell}$. We also hope that the number of ``rbb-triangle-like'' structures is still small; in particular, we will still get a similar bound $k^{2 + o(1)}$ for such generalized structure, similar to our previous bound for the red-blue-blue triangles. When this is the case, we need $kd^\ell \geqs k^{2 + o(1)}$, meaning that when $\ell = \omega(1)$ it suffices to select $d = k^{o(1)}$, which yields $k^{1 - o(1)}$ factor inapproximability as desired. To facilitate our discussion, let us define notations for $\ell$-walks here.

\begin{definition}[$\ell$-Walks]
For any red/blue graph $G = (V, E_r \cup E_b)$ and any integer $\ell \geqs 2$, an \emph{$\ell$-blue-walk} in $G$ is an $(\ell + 1)$-tuple of vertices $(v_1, v_2, \dots, v_{\ell + 1}) \in V^{\ell + 1}$ such that every pair of consecutive vertices are joined by a blue edge, i.e., $\{v_i, v_{i + 1}\} \in E_b$ for every $i \in [\ell]$. For brevity, we sometimes refer to $\ell$-blue walks simply as \emph{$\ell$-walks}. We use $\cW^G_\ell$ to denote the set of all $\ell$-walks in $G$.
\end{definition}

Note here that a vertex can appears multiple times in a single $\ell$-walk.

One detail we have yet to specify in the proof is the structure that generalizes the rbb triangle for $\ell$-walks where $\ell > 2$. Like before, this structure will enforce the two end points of the walk to be joined by a red edge, i.e., $\{v_1, v_{\ell + 1}\} \in E_r$. Additionally, we require every pair of non-consecutive vertices to be joined by a red edge. We call such a walk a \emph{red-filled $\ell$-walk} (see Figure~\ref{fig:4-walks}):

\begin{definition}[Red-Filled $\ell$-Walks]
For any red/blue graph $G = (V, E_r \cup E_b)$, a \emph{red-filled $\ell$-walk} is an $\ell$-walk $(v_1, v_2, \dots, v_{\ell + 1})$ such that every pair of non-consecutive vertices is joined by a red edge, i.e., $\{v_i, v_j\} \in E_r$ for every $i, j \in [\ell + 1]$ such that $j > i + 1$. Let $\tcW^G_\ell$ denote the set of all red-filled $\ell$-walks in $G$. Moreover, for every $u, v \in V$, let $\tcW^G_\ell(u, v)$ denote the set of all red-filled $\ell$-walks from $u$ to $v$, i.e., $\cW^G_\ell(u, v) = \{(v_1, \dots, v_{\ell + 1}) \in \tcW^G_\ell \mid v_1 = u \wedge v_{\ell + 1} = v\}$.
\end{definition}

As mentioned earlier, we will need a generalized transitivity property that works not only for rbb triangles but also for our new structure, i.e. the red-filled $\ell$-walks. This can be defined analogously to $q$-red/blue transitivity as follows.

\begin{definition}[$(q, \ell)$-Red/Blue Transitivity]
For any positive integers $q, \ell \in \N$, a red/blue graph $G = (V, E_r \cup E_b)$ is said to be \emph{$(q, \ell)$-red/blue-transitive} if, for every pair of vertices $u, v \in V$ that are joined by a red edge, there exists at most $q$ red-filled $\ell$-walks starting at $u$ and ending at $v$, i.e., $|\tcW^G_\ell(u, v)| \leqs q$.
\end{definition}

Using a similar argument to before, we can show that, when $\cS$ consists of random subsets where each element is included in a subset with probability $\alpha$, the two-level agreement graph is $(q, \ell)$-red/blue transitive for some parameter $q$ that is a function of only $\alpha$ and $\ell$. When $1/\alpha$ and $\ell$ are small enough in terms of $k$, $q$ can made to be $k^{o(1)}$. (The full proof can be found in Section~\ref{subsec:transitive-1}.)

Once this is proved, it is not hard (using a similar argument as before) to show that, when $d \gg (kq)^{1/\ell}$, most $\ell$-walks are not red-filled, i.e., $|\cW^G_\ell| \gg |\tcW^G_\ell|$. Even with this, it is still unclear how we can get back a ``clique-like'' subgraph; in the case of $\ell = 2$ above, this implies that a blue-neighborhood induces few red edges, but the argument does not seem to generalize to larger $\ell$. Fortunately, it is still quite easy to find a large subgraph that a non-trivial fraction of pairs of vertices do \emph{not} form red edges; specifically, we will find two subsets $U_1, U_2 \subseteq V$ each of size $d$ such that for at least $1/\ell^2$ fraction of $(u_1, u_2) \in U_1 \times U_2$, $\{u_1, u_2\}$ is not a red edge. To find such sets, observe that, if $|\cW^G_\ell| \geqs 2|\tcW^G_\ell|$, then for a random $(v_1, \dots, v_{\ell + 1}) \in \cW^G_\ell$ the probability that there exists non-consecutive vertex $v_i, v_j$ in the walk that are joined by a red edge is at least $1/2$. Since there are less than $\ell^2/2$ such $i, j$, union bound implies that there must be non-consecutive $i^*, j^*$ such that the probability that $v_{i^*}, v_{j^*}$ are not joined by a red edge is at least $1/\ell^2$. Let us assume without loss of generality that $i^* < j^*$; since they are not consecutive, we have $i^* + 1 < j^*$.

Let us consider $v_{i^* + 1}, v_{j^* - 1}$. By a simple averaging argument, there must be $u^*$ and $w^*$ such that, conditioning on $v_{i^* + 1} = u^*$ and $v_{j^* + 1} = w^*$, the probability that $\{v_{i^*}, v_{j^*}\} \notin E_r$ is at least $1/\ell^2$. However, this conditional probability is exactly equal to fraction of $(u_1, u_2) \in N_b(u^*) \times N_b(w^*)$ that $u_1$ and $u_2$ are not joined by a red edge. Recall again that $N_b(v)$ is used to denote the set of all blue-neighbors of $v$. Thus, $U_1 = N_b(u^*)$ and $U_2 = N_b(w^*)$ are the sets with desired property.

We are still not done yet since we have to use these sets to decode back the global function $g$. This is still not obvious: the guarantee we have for our sets $U_1, U_2$ is rather weak since we only know that at least $1/\ell^2$ of the pairs of vertices from the two sets do not form red edges. This is in contrast to the $\ell = 2$ case where we have a subgraph such that almost all induced edges are \emph{not} red. 

To see how to overcome this barrier, recall that a pair $S_1, S_2$ that does not form a red edge corresponds to $f_{S_1} \overset{\zeta}{\approx} f_{S_2}$. As a thought experiment, let us think of the following scenario: if instead of just $\zeta$-consistency, these pairs satisfy (exact) consistency, then we can consider the collection $\tcF = \{f_{S}\}_{S \in \tU}$ where $\tU = U_1 \cup U_2$. This is a collection of $\Theta(d)$ local functions such that $\agr(\tcF) \geqs \Omega(1/\ell^2)$. Thus, when $d \gg \ell^4$, we are in the regime where $\agr(\tcF) \gg 1/d^{1/2}$, meaning that we can apply our earlier argument (for the $\delta \geqs k^{o(1)}/k^{1/2}$ regime) to recover $g$!

The approach in the previous paragraph of course does not work directly because we only know that $\Omega(1/\ell^2)$ fraction of the pairs $\{S_1, S_2\} \subseteq \tU$ are $\zeta$-consistent, not exactly consistent. However, we can still try to mimic the proof in the regime $\delta \geqs k^{o(1)}/k^{1/2}$ and define a red/blue graph in such a way that such $\zeta$-consistent pairs are now blue edges. Naturally, the red edges will now be the $\zeta'$-inconsistent pairs for some $\zeta' > \zeta$. In other words, we consider the \emph{generalized two-level consistency graph} defined as follows.

\begin{definition}[Generalized Two-Level Consistency Graph]
Given a collection of functions $\cF = \{f_S\}_{S \in \cS}$ and two real numbers $0 \leqs \zeta \leqs \zeta' \leqs 1$, the generalized two-level consistency graph $G^{\cF, \zeta, \zeta'} = (V^{\cF, \zeta, \zeta'}, E^{\cF, \zeta, \zeta'}_r \cup E^{\cF, \zeta, \zeta'}_b)$ is a red/blue graph defined as follows.
\begin{itemize}
\item The vertex set $V^{\cF, \zeta, \zeta'}$ is simply $\cS$.
\item The blue edges are the $\zeta$-consistent pairs $\{S_1, S_2\}$, i.e., $E^{\cF, \zeta, \zeta'}_b = \{\{S_1, S_2\} \in \binom{\cS}{2} \mid f_{S_1} \overset{\zeta}{\approx} f_{S_2}\}$.
\item The red edges are the $\zeta'$-inconsistent pairs $\{S_1, S_2\}$, i.e., $E^{\cF, \zeta, \zeta'}_r = \{\{S_1, S_2\} \in \binom{\cS}{2} \mid f_{S_1} \overset{\zeta'}{\napprox} f_{S_2}\}$.
\end{itemize}
\end{definition}

As its name suggests, the generalized two-level consistency graph is a generalization of the two-level consistency graph from Definition~\ref{def:two-lev}; namely $G^{\cF, 0, \zeta}$ in the more general definition coincides with $G^{\cF, \zeta}$ in the original definition.

Now, it is not hard to show that when $\zeta' \gg \zeta / \alpha$, the graph $G^{\cF, \zeta, \zeta'}$ is again $q$-red/blue transitive for some $q$ that depends only on $\alpha$ and $\zeta$. This means that we can apply our argument from the $\delta \geqs 1/k^{1/2 - o(1)}$ regime on the graph $G^{\tcF,\zeta,\zeta'}$, which yields a subset $U \subseteq \tU$ such that almost all pairs $\{S_1, S_2\} \subseteq U$ are $\zeta'$-consistent. By selecting the parameters appropriately, such an almost non-red subgraph once again gives us the desired global function. This wraps up our proof overview.

{\bf Changes from Previous Version.} The proof presented in this manuscript differs slightly from that in the conference version of this work~\cite{DM18}. Specifically, while both versions follow the same approach to find $U_1, U_2$, they diverge afterwards. In~\cite{DM18}, instead of reapplying the argument of the regime $\delta \gg 1/k^{1/2}$ on $U_1, U_2$ as we do in this version, we resorted to the K\H{o}v\'{a}ri-S\'{o}s-Tur\'{a}n (KST) Theorem~\cite{KST54}, which roughly states that every dense bipartite graph has a biclique (complete bipartite subgraph) of logarithmic size. By applying this theorem to the bipartite graph between $U_1$ and $U_2$ where there is an edge between $u_1 \in U_1$ and $u_2 \in U_2$ if and only if $\{u_1, u_2\}$ is not a red edge, we arrived at subsets $V_1 \subseteq U_1$, $V_2 \subseteq U_2$ of reasonable large size such that for all $(u_1, u_2) \in V_1 \times V_2$, $u_1$ and $u_2$ are not joined by a red edge. Finally, we observed that such a ``non-red biclique'' can also be used to decode a global function, by taking the majority of either $V_1$ or $V_2$ side.

A disadvantage of the proof in~\cite{DM18} is that the agreement theorem there finds a global function that approximately agrees with only $\Omega_{\beta, \alpha}(\log (\delta k))$ of the local functions whereas Theorem~\ref{thm:agr-informal} finds one that approximately agrees with $\delta k^{1 - o_{\beta, \alpha}(1)}$ local functions, which, as discussed earlier, is nearly optimal. This loss in~\cite{DM18} also affects the low order term in the inapproximability ratio: while the ETH-hardness in the current manuscript has inapproximability ratio $k/2^{(\log k)^{1/2 + \rho}}$ for any $\rho > 0$, the ratio in~\cite{DM18} is only $k/2^{(\log k)^{1 - \gamma}}$ for some (small) $\gamma > 0$.

\section{Preliminaries} \label{sec:prelim}

%\subsection{Constraint Satisfaction Problems and 3-SAT}
%\paragraph{2-CSPs.} An instance $\Gamma$ of 2-CSP consists of
%\begin{itemize}
%\item a constraint graph $G = (V, E)$,
%\item for each vertex $v \in V$, an alphabet set $\Sigma_v$, and,
%\item for each edge $\{u, v\} \in E$, a constraint $C_{uv} \subseteq \Sigma_u \times \Sigma_v$.
%\end{itemize}
%A labeling of $\Gamma$ is a tuple $(\sigma_v)_{v \in V}$ such that $\sigma_v \in \Sigma_v$ for every $v \in V$. An edge $(u, v) \in E$ is said to be \emph{satisfied} by the labeling if $(\sigma_u, \sigma_v) \in C_{uv}$. The value of a labeling $\sigma = (\sigma_v)_{v \in V}$, denoted by $\val(\sigma)$, is defined as the fraction of edges that it satisfies, i.e., $\val(\sigma) = |\{\{u, v\} \in E \mid (\sigma_u, \sigma_v) \in C_{uv}\}| / |E|$. The goal is to find $\sigma$ with maximum value; we denote the such optimal value by $\val(\Gamma)$, i.e., $\val(\Gamma) = \max_{\sigma} \val(\sigma)$.

%\paragraph{3-SAT.} An instance $\Phi$ of 3-SAT consist of a variable set $\cX$ and the clause set $\cC$.  For each clause $C \in \cC$, we use $\var(C)$ to denote the variables whose literals appear in $C$. We also extend this notation naturally to sets of clauses, i.e., for every $S \subseteq \cC$, $\var(S) = \bigcup_{C \in S} \var(S)$.

\subsection{Well-Behaved Subsets}

We next define two properties of collections of subsets, which will be needed in our soundness analysis. First, recall that, in our proof overview for the weaker $k^{1/2 - o(1)}$ factor hardness, we need the following to show the red/blue transitivity of the consistency graph: for any $r$ subsets from the collection, their union must contain almost all clauses. Here $r$ is a positive integer that effects the red/blue transitivity parameter. Collections with this property are sometimes called \emph{dispersers}. For walks with larger length, we need a stronger property that any union of $r$ intersections of $\ell$ subsets are large. We call such collections \emph{intersection dispersers}:

\begin{definition}[Intersection Disperser]
For any universe $\cU$, a collection $\cS$ of subsets of $\cU$ is an \emph{$(r, \ell, \eta)$-intersection disperser} if, for any $r$ disjoint subcollections $\cS^{1}, \dots, \cS^{r} \subseteq \cS$ each of size at most $\ell$, we have
\begin{align*}
\left|\bigcup_{i=1}^{r}\left(\bigcap_{S \in \cS^{i}} S\right)\right| \geqs (1 - \eta) |\cU|.
\end{align*}
\end{definition}

Note that in the definition we require $\cS^{1}, \dots, \cS^{r}$ to be disjoint. This is necessary because otherwise we can include a common set $S \in \cS$ into all the subcollections. In this case, the union will be contained in $S$ and hence will not cover almost all the universe.

Another property we need is that any sufficiently large subcollection $\tcS$ of $\cS$ is ``sufficiently uniform''. This is used when we decode a good assignment from an almost non-red subgraph. More specifically, the uniformity condition requires that almost all clauses appear in not too small number of subsets in $\tcS$, as formalized below.

\begin{definition}[Uniformity]
For any universe $\cU$, a collection $\tcS$ of subsets of $\cU$ is $(\gamma, \mu)$-uniform if, for at least $(1 - \mu)$ fraction of elements $u \in \cU$, $u$ appears in at least $\gamma$ fraction of the subsets in $\tcS$. In other words, $\tcS$ is $(\gamma, \mu)$-uniform if and only if $|\{u \in \cU \mid \Pr_{S \in \tcS}[u \in S] \geqs \gamma\}| \geqs (1 - \mu)|\cU|$.
\end{definition}

Using standard concentration bounds, it is not hard to show that, when $m$ is sufficiently large, a collection of random subsets where each element is included in each subset independently with probability $\alpha$ is an $(1/O(\alpha^\ell), \ell, O(1))$-disperser and every subcollection of size $\Omega(1/\alpha)$ is $(\alpha, O(1))$ uniform. The exact parameter dependencies are shown in the lemma below.

\begin{lemma}[Deterministic Construction of Well-Behaved Subsets] \label{lem:well-behaved-set-deterministic}
For any $0 < \alpha, \mu, \eta < 1$ and any $k, \ell \in \N$, let $m_0$ be $1000(\log k \log(1/\mu)/(\alpha\mu^2) + \ell \log(1/\eta)\log k / (\alpha^\ell\eta) + 1/\alpha + 1)$. For any integer $m \geqs m_0$ and $m$-element set $\cU$, there is a collection $\cT$ of subsets of $\cU$ with the following properties.
\begin{itemize}
\item (Size) Every subset in $\cT$ has size at most $2\alpha m$.
\item (Intersection Disperser) $\cT$ is a $(\lceil \ln(2/\eta)/\alpha^\ell\rceil, \ell, \eta)$-disperser.
\item (Uniformity) Any subcollection $\tcT \subseteq \cT$ of size $\lceil 8\ln(2/\mu)/\alpha \rceil$ is $(\alpha/2, \mu)$-uniform.
\end{itemize}
Moreover, such a collection $\cT$ can be deterministically constructed in time $\poly(m)2^{O(m_0 k^2)}$.
\end{lemma}

Since all techniques involved in the proof of Lemma~\ref{lem:well-behaved-set-deterministic} are standard, we defer it to Appendix~\ref{app:well-behaved-sets}. Let us turn our focus back to our main technical contribution: the agreement testing theorem.

\section{The Main Agreement Theorem} \label{sec:soundness}

The main goal of this section is to prove the following agreement theorem, which is the formal version of Theorem~\ref{thm:agr-informal} and is also the main technical contribution of this work.

\begin{theorem} \label{thm:agr-main-formal}
For any $0 < \eta, \zeta, \gamma, \mu < 1$ and $r, \ell, k, h, n, d \in \N$ such that $\ell \geqs 2$, let $\cS$ be any collection of $k$ subsets of $[n]$ such that $\cS$ is $(r, \ell, \zeta)$-intersection disperser and every subcollection $\tcS \subseteq \cS$ of size $h$ is $(\gamma, \mu)$-uniform, and let $\cF = \{f_S\}_{S \in \cS}$ be any collection of functions. If $\delta \triangleq \agr(\cF) \geqs \frac{10 + 64(r\ell)^2 k^{1/\ell}}{k}$, then there exists a subcollection $\cS' \subseteq \cS$ of size at least $\frac{\delta k}{256 \ell^2}$ and a function $g: [n] \to \{0, 1\}$ such that $g \overset{\beta}{\approx} f_S$ for all $S \in \cS'$ where $$ \beta = 2\sqrt{\frac{65536 h \ell^6}{\delta k} + \mu + 2 \zeta / \gamma}.$$
\end{theorem}

While the parameters of the theorem can be confusing, when each subset in $\cS$ is a random $\alpha n$-size subset of $[n]$, the parameters we are interested in are as follows: $\mu$ and $\eta$ both go to $0$ as $n$ goes to infinity, $h$ and $\gamma$ depend only on $\alpha$, and, $r$ is $O(1/\alpha^\ell)$. Since we want the requirement on soundness as weak as possible, we want to minimize $(r\ell)^2k^{1/\ell} = 2^{O_\alpha(\ell + (\log k) / \ell)}$. Hence, our best choice is to let $\ell = \sqrt{\log k}$, which indeed yields the $k/2^{(\log k)^{1/2 + \rho}}$ ratio inapproximability for 2-CSPs.

To prove this theorem, we follow the general outline as stated in the proof overview section. In particular, the proof contains five main steps, as elaborated below.
\begin{enumerate}[(1)]
\item First, we will show that when $\cS$ is an intersection disperser with appropriate parameters, then the two-level consistency graph $G^{\cF, \zeta}$ satisfies $(q, \ell)$-red/blue transitivity for certain $q, \ell$. \label{step:transitivity-1}
\item Second, we argue that, for any red/blue transitive graphs that contains sufficiently many blue edges, we can find a large subset $\tU$ of vertices such that a reasonably large fraction of pairs $\{S_1, S_2\} \subseteq \tU$ are non-red. This is done by counting red-filled $\ell$-walks for an appropriate $\ell$. \label{step:subgraph-1}
\item We then focus on $\tcF = \{f_S\}_{\tU}$ and show, using a uniformity condition of $\cS$, that the generalized two-level consistency graph $G^{\tcF, \zeta, \zeta'}$ is red/blue transitive with certain parameters. \label{step:transitivity-2}
\item Next, counting rbb triangles reveals a large ``almost non-red subgraph'' in the graph $G^{\tcF, \zeta, \zeta'}$. \label{step:subgraph-2}
\item Finally, we decode a global function from this almost non-red subgraph. \label{step:decode}
\end{enumerate}

This section is organized as follows. In Subsection~\ref{sec:transitive}, we show transitivity properties of the two-level and generalized two-level consistency graphs, i.e., Steps~\ref{step:transitivity-1} and~\ref{step:transitivity-2}. Subsection~\ref{sec:subgraph} contains a structural lemma regarding an existence of a large subgraph with certain non-red density in red/blue transitive graphs; this lemma is at the heart of Steps~\ref{step:subgraph-1} and~\ref{step:subgraph-2}. Next, in Subsection~\ref{sec:maj-decode}, we prove Step~\ref{step:decode}. Finally, in Subsection~\ref{sec:agr-proof}, we put these parts together and prove Theorem~\ref{thm:agr-main-formal}.

\subsection{Red/Blue-Transitivity of (Generalized) Two-Level Consistency Graph} \label{sec:transitive}

%In this section, we prove red/blue transitivity of two-level consistency graph and generalized two-level consistency graph, using the assumptions that $\cS$ is an intersection disperser and that every large subcollection of $\cS$ is sufficiently uniform respectively.

\subsubsection{Red/Blue-Transitivity from Intersection Disperser}
\label{subsec:transitive-1}

The first step in our proof is to show that the two-level consistency graph $G^{\cF, \zeta}$ is red/blue-transitive, assuming that $\cS$ is an intersection disperser. Specifically, our main lemma is the following:

\begin{lemma} \label{lem:transitivity}
If $\cS$ is an $(r, \ell, \zeta)$-intersection disperser, then, for any $\cF = \{f_S\}_{S \in \cS}$, $G^{\cF, \zeta}$ is $((r\ell)^{2(\ell - 1)}, \ell)$-red/blue-transitive.
\end{lemma}

We note here that both in Lemma~\ref{lem:transitivity} and Claim~\ref{claim:disjoint-transitivity} below, the transitivity property holds not only for $\ell$-walks as specified in the statements, but also for $(\ell + 1)$-walks. However, since the latter does not yield any improvement to our main results, we work with only $\ell$-walks, which makes the calculations cleaner.

In other words, we would like to show that, for every $S_1, S_2 \in \cS$ that are joined by a red edge in $G^{\cF, \zeta}$, there are at most $(r\ell)^{2(\ell - 1)}$ red-filled $\ell$-walks from $S_1$ to $S_2$. The intersection disperser does not immediately imply such a bound, due to the requirement in the definition that the subcollections are disjoint. Rather, it only directly implies a bound on number of \emph{disjoint} $\ell$-walks from $S_1$ to $S_2$, where two $\ell$ walks from $S_1$ to $S_2$, $(T_1 = S_1, \dots, T_{\ell + 1} = S_2), (T'_1 = S_1, \dots, T'_{\ell + 1} = S_2) \in \cW^{G^{\cF, \zeta}}_\ell(S_1, S_2)$, are said to be \emph{disjoint} if they do not share any vertex except the starting and ending vertices, i.e., $\{T_2, \dots, T_{\ell}\} \cap \{T'_2, \dots, T'_{\ell}\} = \emptyset$. Note that multiple walks sharing starting and ending vertices are said to be disjoint if they are mutually disjoint. The following claim follows almost immediately from definition of intersection dispersers:

\begin{claim} \label{claim:disjoint-transitivity}
If $\cS$ is an $(r, \ell, \zeta)$-intersection disperser, then, for any $\cF = \{f_S\}_{S \in \cS}$, any integer $2 \leqs p \leqs \ell$ and any $\{S_1, S_2\} \in E^{\cF, \zeta}_r$, there are less than $r$ disjoint $p$-walks from $S_1$ to $S_2$ in $G^{\cF, \zeta}$.
\end{claim}

\begin{proof}
Suppose for the sake of contradiction that $\cS$ is an $(r, \ell, \zeta)$-intersection disperser but there exist $\cF = \{f_S\}_{S \in \cS}, 2 \leqs p \leqs \ell$ and $\{S_1, S_2\} \in E^{\cF, \zeta}_r$ such that there are at least $r$ disjoint $p$-walks from $S_1$ to $S_2$. Let these walks be $(T_{1, 1} =S_1, T_{1, 2}, \dots, T_{1, p}, T_{1, p + 1} = S_2), \dots, (T_{r, 1} =S_1, T_{r, 2}, \dots, T_{r, p}, T_{r, p + 1} = S_2) \in \cW^{G^{\cF, \zeta}}_p(S_1, S_2)$.

For each $i \in [r]$, consider any $x \in \bigcap_{j = 1}^{p + 1} T_{i, j}$. Apriori this intersection may be empty but since $\cS$ is an intersection disperser this usually does not occur. Since $\{T_{i, j}, T_{i, j + 1}\} \in E^{\cF, \zeta}_b$ for every $j \in [p]$, we have
\begin{align*}
f_{S_1}(x) = f_{T_{i, 1}}(x) = f_{T_{i, 2}}(x) = \cdots = f_{T_{i, p}}(x) = f_{T_{i, p + 1}}(x) = f_{S_2}(x).
\end{align*}

Hence, for every $x \in \bigcup_{i = 1}^r \left(\bigcap_{j = 1}^{p + 1} T_{q, j}\right)$, $f_{S_1}(x) = f_{S_2}(x)$. Let $T^*$ denote $\bigcup_{i = 1}^r \left(\bigcap_{j = 2}^{p} T_{q, j}\right)$. Since $T_{i, 1} = S_1$ and $T_{i, p + 1} = S_2$ for all $i \in [r]$, we have
\begin{align*}
\bigcup_{i = 1}^r \left(\bigcap_{j = 1}^{p + 1} T_{q, j}\right) = (S_1 \cap S_2) \cap T^*.
\end{align*}
In other words, $f_{S_1}$ and $f_{S_2}$ can only disagree on variables outside of $T^*$. However, since $\cS$ is an $(r, \ell, \zeta)$-intersection disperser, we have $|T^*| \geqs (1 - \zeta)n$. Hence, $\disa(S_1, S_2) \leqs \zeta n$, which contradicts with $\{S_1, S_2\} \in E^{\cF, \zeta}_r$.
\end{proof}

Since all 2-walks from $S_1$ to $S_2$ are disjoint, the above claim immediately gives a bound on the number of red-filled 2-walks from $S_1$ to $S_2$. To bound the number of red-filled walks of larger lengths, we will use induction on the length of the walks. Suppose that we have bounded the number of red-filled $i$-walks sharing starting and ending vertices for $i \leqs z - 1$. The key idea in the proof is that we can use this inductive hypothesis to show that, for any $S_1, S_2, S \in \cS$, few $z$-walks from $S_1$ to $S_2$ contain a given $S$. Here we say that a $z$-walk $(T_1 = S_1, \dots, T_z = S_2)$ from $S_1$ to $S_2$ \emph{contains} $S$ if $S \in \{T_2, \dots, T_z\}$. This implies that for a given $z$-walk from $S_1$ to $S_2$ there are only few walks that are not disjoint from it. This allows us to show that, if there are too many $z$-walks, then there must also be many disjoint $z$-walks as well, which would violate Claim~\ref{claim:disjoint-transitivity}. A formal proof of Lemma~\ref{lem:transitivity} based on this intuition is given below.

\begin{proof}[Proof of Lemma~\ref{lem:transitivity}]
For every integer $i$ such that $2 \leqs i \leqs \ell$, let $P(i)$ denote the following statement: for every $S_1, S_2 \in \cS$, $|\tcW^{G^{\cF, \zeta}}_i(S_1, S_2)| \leqs (ri)^{2(i - 1)}$. For convenience, let $B_i = (ri)^{2(i - 1)}$ for every $2 \leqs i \leqs \ell$.

{\bf Base Case.} Since every different $2$-walks from $S_1$ to $S_2$ are disjoint, Claim~\ref{claim:disjoint-transitivity} immediately implies that the number of $2$-walks from $S_1$ to $S_2$ is at most $r \leqs B_2$.

{\bf Inductive Step.}

Suppose that, for some integer $z$ such that $3 \leqs z \leqs \ell$, $P(3), \dots, P(z - 1)$ are true. We will show that $P(z)$ is true. To do so, let us first prove that, for any fixed starting and ending vertices, any vertex cannot appears in too many red-filled $z$-walks, as stated in the following claim.

\begin{claim}
For all $S_1, S_2, S \in \cS$, the number of red-filled $z$-walks from $S_1$ to $S_2$ containing $S$ in $G^{\cF, \zeta}$ is at most $B_z/(zr)$.
\end{claim}

\begin{subproof}
First, observe that the number of red-filled $z$-walks from $S_1$ to $S_2$ containing $S$ is at most the sum over all positions $2 \leqs j \leqs z$ of the number of $z$-walks from $S_1$ to $S_2$ such that the $j$-th vertex in the walk is $S$. More formally, the number of red-filled $z$-walks from $S_1$ to $S_2$ containing $S$ is
\begin{align*}
|\{(T_1, \dots, T_{z + 1}) \in \tcW^{G^{\cF, \zeta}}_z(S_1, S_2) \mid \exists 2 \leqs j \leqs z, T_j = S_j\}| \leqs \sum_{j=2}^{z} |\{(T_1, \dots, T_{z + 1}) \in \tcW^{G^{\cF, \zeta}}_z(S_1, S_2) \mid T_j = S\}|.
\end{align*}

Now, for each $2 \leqs j \leqs z$, to bound the number of red-filled $z$-walks from $S_1$ to $S_2$ whose $j$-th vertex is $S$, let us consider the following three cases based on the value of $j$:
\begin{enumerate}
\item $3 \leqs j \leqs z - 1$. Observe that, for any such walk $(T_1 = S_1, T_2, \dots, T_j = S, \dots, T_z, T_{z + 1} = S_2)$, the subwalk $(T_1 = S_1, \dots, T_j = S)$ and $(T_j = S, \dots, T_{z + 1} = S_2)$ must be red-filled walks as well. Since the numbers of red-filled $(j - 1)$-walks from $S_1$ to $S$ and red-filled $(z + 1 - j)$-walks from $S$ to $S_2$ are bounded by $B_{j - 1}$ and $B_{z + 1 - j}$ respectively (from the inductive hypothesis), there are at most $B_{j - 1}$ choices of $(T_1 = S_1, \dots, T_j = S)$ and $B_{z + 1 - j}$ choices of $(T_j = S, \dots, T_{z - 1}, T_z = S_2)$. Hence, there are at most $B_{j - 1}B_{z + 1 - j}$ red-filled $z$-walks from $S_1$ to $S_2$ whose $j$-th vertex is $S$.
\item $j = 2$. In this case, the subwalk $(T_j = S, \dots, T_{z + 1} = S_2)$ must be a red-filled $(z - 1)$-walk from $S$ to $S_2$. Hence, the number of red-filled $z$-walks from $S_1$ to $S_2$ where $T_j = S$ is bounded above by $B_{z - 1}$.
\item $j = z$. Similar to the previous case, we also have the bound of $B_{z - 1}$.
\end{enumerate}
For convenience, let $B_1 = 1$. The above argument gives us the following bound for every $2 \leqs j \leqs z$:
\begin{align*}
|\{(T_1, \dots, T_{z + 1}) \in \tcW^{G^{\cF, \zeta}}_z(S_1, S_2) \mid T_j = S\}| &\leqs B_{j - 1}B_{z + 1 - j}.
\end{align*}

Summing this over $j$, we have the following upper bound on the number of red-filled $z$-walks from $S_1$ to $S_2$ containing $S$:
\begin{align*}
\sum_{j=2}^{z} B_{j - 1} B_{z + 1 - j} = \sum_{j=2}^{z} (r(j - 1))^{2(j-2)}(r(z + 1 - j))^{2(z - j)} \leqs \sum_{j=2}^{z} (rz)^{2(z - 2)} \leqs B_z / (zr),
\end{align*}
which concludes the proof of the claim.
\end{subproof}

Having proved the above claim, it is now easy to show that $P(z)$ is true. Suppose for the sake of contradiction that there exists $S_1, S_2 \in \cS$ such that $|\tcW^{G^{\cF, \zeta}}_z(S_1, S_2)| > B_z$. Consider the following procedure of selecting disjoint walks from $\tcW^{G^{\cF, \zeta}}_z(S_1, S_2)$. First, initialize $U = \tcW^{G^{\cF, \zeta}}_z(S_1, S_2)$ and repeat the following process as long as $U \ne \emptyset$: select any $(T_1, \dots, T_{z + 1}) \in U$ and remove every $(T'_1, \dots, T'_{z + 1})$ that is not disjoint with $(T_1, \dots, T_{z + 1})$ from $U$. Observe that, each time a walk $(T_1, \dots, T_{z + 1})$ is selected, the number of walks removed from $U$ is at most $B_z/r$; this is because each removed walk must contain at least one of $T_2, \dots, T_z$, but, from the above claim, each of these vertices are contained in at most $B_z/(zr)$ walks. Since we start with more than $B_z$ walks, at least $r$ walks are picked. These walks are disjoint $z$-walks starting from $S_1$ and $S_2$, which, due to Claim~\ref{claim:disjoint-transitivity}, is a contradiction. Thus, $P(z)$ is true as desired.

Hence, $P(\ell)$ is true, which, by definition, implies that $G_\cF$ is $((r\ell)^{2(\ell - 1)}, \ell)$-red/blue-transitive.
\end{proof}

\subsubsection{Red/Blue-Transitivity from Uniformity}

In Step~\ref{step:transitivity-2} of our proof, we need to show red/blue-transitivity of the generalized two-level consistency graph $G^{\cF, \zeta, \zeta'}$. This is encapsulated in the following lemma.

\begin{lemma} \label{lem:transitivity-2}
If every subcollection $\tcS \subseteq \cS$ of size $r$ is $(\gamma, \mu)$-uniform, then, for any $\zeta \geqs 0$, $\zeta' \geqs \mu + 2 \zeta / \gamma$ and $\cF = \{f_S\}_{S \in \cS}$, the generalized two-level consistency graph $G^{\cF, \zeta, \zeta'}$ is $r$-red/blue transitive. 
\end{lemma}

The proof of the lemma is quite simple. The key observation is that, if $S_1$ and $S_2$ are joined by a red edge and $T$ is a common blue-neighbor in the graph $G^{\cF, \zeta, \zeta'}$, then it means that $T$ only hits a small number (i.e. $2\zeta n$) of the variables on which $f_{S_1}$ and $f_{S_2}$ disagree. In other words, such variables appear less frequently in common blue-neighbors of $S_1$ and $S_2$. If the common-blue neighbor set is of size $r$, this contradicts the fact that the set is $(\gamma, \mu)$-uniform. This intuition is formalized below.

\begin{proof}[Proof of Lemma~\ref{lem:transitivity-2}]
Suppose for the sake of contradiction that $G^{\cF, \zeta, \zeta'}$ is not $r$-red/blue transitive. That is, there exist $S_1, S_2 \in \cS$ that are joined by a red edge such that there are $r$ red-filled $2$-walks (i.e. rbb triangle) from $S_1$ to $S_2$. Suppose that these walks are $(S_1, T_1, S_2), (S_1, T_2, S_2), \cdots, (S_1, T_r, S_2)$.

For every $i \in [r]$, since $(S_1, T_i, S_2)$ is a 2-walk, $\{S_1, T_i\}$ and $\{S_2, T_i\}$ are blue edges. This implies that
\begin{align}
\disa(f_{S_1}, f_{T_i}), \disa(f_{S_2}, f_{T_i}) \leqs \zeta n. \label{eq:tran-1}
\end{align}

On the other hand, we can lower bound $\EX_{i \in [r]}[\disa(f_{S_1}, f_{T_i}) + \disa(f_{S_2}, f_{T_i})]$ as follows. First, let let $\cX^{\disa}$ denote the set of all $x \in S_1 \cap S_2$ such that $f_{S_1}(x) \ne f_{S_2}(x)$; since $\{S_1, S_2\}$ is a red edge, we have $|\cX^{\disa}| > \zeta' n$. We can rearrange $\EX_{i \in [r]}[\disa(f_{S_1}, f_{T_i}) + \disa(f_{S_2}, f_{T_i})]$ as
\begin{align}
&\EX_{i \in [r]}[\disa(f_{S_1}, f_{T_i}) + \disa(f_{S_2}, f_{T_i})] \nonumber \\
&= \sum_{x \in [n]} \left(\Pr_{i \in [r]}[x \in (S_1 \cap T_i ) \wedge f_{S_1}(x) \ne f_{T_i}(x)] + \Pr_{i \in [r]}[x \in (S_2 \cap T_i) \wedge f_{S_2}(x) \ne f_{T_i}(x)]\right) \nonumber \\
&\geqs \sum_{x \in \cX^{\disa}} \left(\Pr_{i \in [r]}[x \in T_i \wedge f_{S_1}(x) \ne f_{T_i}(x)] + \Pr_{i \in [r]}[x \in T_i \wedge f_{S_2}(x) \ne f_{T_i}(x)]\right) \nonumber \\
&\geqs \sum_{x \in \cX^{\disa}} \left(\Pr_{i \in [r]}[x \in T_i \wedge \left(f_{S_1}(x) \ne f_{T_i}(x) \vee f_{S_2}(x) \ne f_{T_i}(x)\right)]\right) \nonumber \\
&= \sum_{x \in \cX^{\disa}} \Pr_{i \in [r]}[x \in T_i] \label{eq:tran-2}
\end{align}
We remark here that the second inequality comes from union bound, whereas the last equality follows from the fact that $(f_{S_1}(x) \ne f_{T_i}(x)) \vee (f_{S_2}(x) \ne f_{T_i}(x))$ is always true when $f_{S_1}(x) \ne f_{S_2}(x)$.

Recall that $\{T_1, \dots, T_r\} \subseteq \cS$ is a subcollection of size $r$ and is thus $(\gamma, \mu)$-uniform. Let $\cX_{\geqs \gamma}$ be the set of all $x \in [n]$ that appears in at least $\gamma$ fraction of $T_i$'s. The $(\gamma, \mu)$-uniformity of $\{T_1, \dots, T_r\}$ implies that $|\cX_{\geqs \gamma}| \geqs (1 - \mu)n$. From this and from $|\cX^{\disa}| > \zeta' n$, we can lower bound the right hand side of (\ref{eq:tran-2}) further as follows:
\begin{align}
\sum_{x \in \cX^{\disa}} \Pr_{T_i \in \cT}[x \in T_i] \geqs \sum_{x \in \cX^{\disa} \cap \cX_{\geqs \gamma}} \Pr_{T_i \in \cT}[x \in T_i] \geqs \gamma |\cX^{\disa} \cap \cX_{\geqs \gamma}| > \gamma (\zeta'  - \mu)n \geqs 2\zeta n \label{eq:tran-3}
\end{align}
where the last inequality comes from our assumption that $\zeta' \geqs \mu + 2 \zeta / \gamma$.

Finally, combining (\ref{eq:tran-1}), (\ref{eq:tran-2}) and (\ref{eq:tran-3}) yields the desired contradiction.
\end{proof}

\subsection{Finding Almost Non-Red Subgraph in Red/Blue-Transitive Graph} \label{sec:subgraph}

%Having finished the first part of the proof, 
%We next move on to the next part of the proof. In Step~\ref{step:subgraph-1}, we would like to show that any $(q, \ell)$-red/blue transitive graph with sufficiently many edges must contain a sufficiently large subgraph whose significant fraction of edges are . contain a sufficiently large almost non-red clique, as stated formally below.
Recall that in two steps of our proofs, we need to utilize the red/blue transitivity of the (generalized) two-level consistency graph to find a large subgraph with certain number of non-red pairs:
\begin{itemize}
\item Specifically, in Step~\ref{step:subgraph-1}, we would like to show that, for appropriate values of $q$ and $\ell$, any $(q, \ell)$-red/blue transitive graph with sufficiently many blue edges must contain a sufficiently large subgraph whose significant (i.e. $1/\ell^2$) fraction of pairs of vertices are non-red.
\item Additionally, in Step~\ref{step:subgraph-2}, we need to show that any $o(d^2/k)$-red/blue transitive graph with sufficiently many blue edges must contain a sufficiently large subgraph such that almost all pairs of its vertices are non-red.
\end{itemize}
It turns out that a single lemma stated below suffices for both steps. In particular, the lemma below returns a subgraph such that roughly $1/\binom{\ell_0}{2}$ fraction of pairs of its vertices are non-red. Plugging in $\ell_0 = \ell$ recovers our former objective whereas setting $\ell_0 = 2$ satisfies the latter. 

\begin{lemma} \label{lem:finding-dense-subgraph}
For every $k_0, q_0, \ell_0, d_0 \in \mathbb{N}$ such that $\ell_0 \geqs 2$ and every $k_0$-vertex $(q_0, \ell_0)$-red/blue-transitive graph $G = (V, E_r \cup E_b)$ such that $|E_b| \geqs 2k_0d_0$, there exist subsets of vertices $U_1, U_2 \subseteq V$ each of size at least $d_0$ such that $|\{(u, v) \in U_1 \times U_2 \mid \{u, v\} \notin E_r\}| \geqs |U_1||U_2|(1 - \frac{q_0k_0}{d_0^{\ell_0}}) / \binom{\ell_0}{2}$. Moreover, when $\ell_0 = 2$, the previous statement remains true even with an additional requirement that $U_1 = U_2$.
\end{lemma}

The proof of Lemma~\ref{lem:finding-dense-subgraph} below is exactly as sketched earlier in Subsection~\ref{sec:overview}.

\begin{proof}[Proof of Lemma~\ref{lem:finding-dense-subgraph}]
We start by preprocessing the graph so that every vertex has blue-degree at least $d_0$. In particular, as long as there exists a vertex $v$ whose blue-degree is at most $d_0$, we remove $v$ from $G$. Let $G' = (V', E'_r \cup E'_b)$ be the graph at the end of this process. Note that we remove less than $k_0d_0$ blue edges in total. Since at the beginning $|E_b| \geqs 2k_0d_0$, we have $|E'_b| \geqs k_0d_0$. Observe also that $G'$ remains $(q_0, \ell_0)$-red/blue-transitive.

Since $V'$ is $(q_0, \ell_0)$-red/blue-transitive, we can bound the number of red-filled $\ell_0$-walk as follows.
\begin{align*}
|\tcW^{G'}_{\ell_0}| = \sum_{u, v \in V' \atop \{u, v\} \in E'_r} |\tcW^{G'}_{\ell_0}(u, v)| \leqs \sum_{u, v \in V' \atop \{u, v\} \in E'_r} q_0 \leqs q_0k_0^2.
\end{align*}

Moreover, notice that $|\cW^{G'}_{\ell_0}| \geqs (k_0d_0) \cdot d_0^{\ell_0 - 1} = k_0d_0^{\ell_0}$; this is simply because there are at least $k_0d_0$ choices for $(v_1, v_2)$ (i.e. all blue edges) and, for any $(v_1, \dots, v_{i - 1})$, there are at least $d_0$ choices for $v_i$.

Hence, we have $|\tcW^{G'}_{\ell_0}|/|\cW^{G'}_{\ell_0}| \leqs q_0k_0/d_0^{\ell_0}$. This implies that $1 - q_0k_0/d_0^{\ell_0} \leqs \Pr_{(v_1, \dots, v_{\ell_0 + 1}) \in \cW^{G'}_{\ell_0}}[(v_1, \dots, v_{\ell_0 + 1}) \notin \tcW^{G'}_{\ell_0}]$. This probability can be further rearranged as follows.
\begin{align*}
\Pr_{(v_1, \dots, v_{\ell_0 + 1}) \in \cW^{G'}_{\ell_0}}[(v_1, \dots, v_{\ell_0 + 1}) \notin \tcW^{G'}_{\ell_0}]
&= \Pr_{(v_1, \dots, v_{\ell_0 + 1}) \in \cW^{G'}_{\ell_0}}[\exists i, j \in [\ell_0 + 1] \text{ such that } j > i + 1, \{v_i, v_j\} \notin E'_r] \\
(\text{Union Bound}) &\leqs \sum_{i, j \in [\ell_0 + 1] \atop j > i + 1} \Pr_{(v_1, \dots, v_{\ell_0 + 1}) \in \cW^{G'}_{\ell_0}}[\{v_i, v_j\} \notin E'_r].
\end{align*}

Now, note that the number of pairs of $i, j \in [\ell_0 + 1]$ such that $j > i + 1$ is $\binom{\ell_0 + 1}{2} - \ell_0 = \binom{\ell_0}{2}$. This implies that there exists one such $i, j$ such that $\Pr_{(v_1, \dots, v_{\ell_0 + 1}) \in \cW^{G'}_{\ell_0}}[\{v_i, v_j\} \notin E'_r] \geqs (1 - \frac{q_0k_0}{d_0^{\ell_0}}) / \binom{\ell_0}{2}$. The probability $\Pr_{(v_1, \dots, v_{\ell_0 + 1}) \in \cW^{G'}_{\ell_0}}[\{v_i, v_j\} \notin E'_r]$ can now be bounded as follows.
\begin{align*}
&\Pr_{(v_1, \dots, v_{\ell_0 + 1}) \in \cW^{G'}_{\ell_0}}[\{v_i, v_j\} \notin E'_r] \\
&= \sum_{u, v} \Pr_{(v_1, \dots, v_{\ell_0 + 1}) \in \cW^{G'}_{\ell_0}}[\{v_i, v_j\} \notin E'_r \mid v_{i + 1} = u \wedge v_{j - 1} = v] \Pr_{(v_1, \dots, v_{\ell_0 + 1}) \in \cW^{G'}_{\ell_0}}[v_{i + 1} = u \wedge v_{j - 1} = v] \\
&\leqs \left(\max_{u, v} \Pr_{(v_1, \dots, v_{\ell_0 + 1}) \in \cW^{G'}_{\ell_0}}[\{v_i, v_j\} \notin E'_r \mid v_{i + 1} = u \wedge v_{j - 1} = v]\right)\left(\sum_{u, v} \Pr_{(v_1, \dots, v_{\ell_0 + 1}) \in \cW^{G'}_{\ell_0}}[v_{i + 1} = u \wedge v_{j - 1} = v]\right) \\
&= \max_{u, v} \Pr_{(v_1, \dots, v_{\ell_0 + 1}) \in \cW^{G'}_{\ell_0}}[\{v_i, v_j\} \notin E'_r \mid v_{i + 1} = u \wedge v_{j - 1} = v]
\end{align*}
where the summation and maximization is taken over all $u, v \in V'$ such that $\Pr_{(v_1, \dots, v_{\ell_0 + 1}) \in \cW^{G'}_{\ell_0}}[v_{i + 1} = u \wedge v_{j - 1} = v]$ is non-zero. Hence, we can conclude that there exists $u^*, v^* \in V'$ such that $$\Pr_{(v_1, \dots, v_{\ell_0 + 1}) \in \cW^{G'}_{\ell_0}}[\{v_i, v_j\} \notin E'_r \mid v_{i + 1} = u^* \wedge v_{j - 1} = v^*] \geqs (1 - \frac{q_0k_0}{d_0^{\ell_0}}) / \binom{\ell_0}{2}.$$

The expression on the left is exactly $|\{(u, v) \in N_b(u^*) \times N_b(v^*) \mid \{u, v\} \notin E'_r\}|/(|N_b(u^*)| \cdot |N_b(v^*)|)$. From this and from every vertex in $G'$ has blue-degree at least $d_0$, $U_1 = N_b(u^*), U_2 = N_b(v^*)$ are the desired sets. Finally, observe that, when $\ell = 2$, we must have $i = 1$ and $j = 3$, resulting in $v_{i + 1} = v_{j - 1}$; this implies that $u^* = v^*$ and we have $U_1 = U_2$.
\end{proof}

\subsection{Majority Decoding of an Almost Non-Red Subgraph} \label{sec:maj-decode}

In the last step of our proof, we will decode a global function $g$ from a sufficiently large almost non-red subgraph in the two-level consistency graph $G^{\cF, \zeta, \zeta'}$. Recall that an almost non-red subgraph in $G^{\cF, \zeta, \zeta'}$ simply corresponds to a subcollection $\cS'$ such that, for almost all pairs $(S_1, S_2) \in \cS' \times \cS'$, $f_{S_1}$ is $\zeta'$-consistent with $f_{S_2}$. The main result of this subsection is that, given such $\cS'$, we can find a global function $g$ that approximately agrees with most of the local functions in the subcollection. This is stated more precisely below.

\begin{lemma} \label{lem:maj-decoding}
Let $\cF = \{f_S\}_{S \in \cS'}$ be a collection of functions such that $\agr_{\zeta'}(\cF) \geqs 1 - \kappa$. Then, the function $g: [n] \to \{0, 1\}$ defined by $g(x) = \maj_{S \in \cS' \atop x \in S} (f_{S}(x))$ satisfies $$\EX_{S \in \cS'}\left[\disa(g, f_S)\right] \leqs n \sqrt{\kappa + \zeta'}.$$
\end{lemma}

\begin{proof}
%As outlined earlier, the function $g$ we take is the majority function of $\cF = \{f_{S}\}_{S \in \cS'}$, i.e., $g(x) = \maj_{S \in \cS' \atop x \in S} (f_{S}(x)).$ Note that here we break tie arbitrarily.
Recall that $\agr_{\zeta'}(\cF) \geqs 1 - \kappa$ is equivalent to $\Pr_{S_1, S_2 \in \cS'}\left[f_{S_1}(x) \overset{\zeta'}{\approx} f_{S_2}(x)\right] \geqs 1 - \kappa$. Hence,
\begin{align}
\EX_{S_1, S_2 \in \cS'}[\disa(f_{S_1}, f_{S_2})] \leqs \Pr_{S_1, S_2 \in \cS'}[f_{S_1} \overset{\zeta'}{\napprox} f_{S_2}] \cdot n + \Pr_{S_1, S_2 \in \cS'}[f_{S_1} \overset{\zeta'}{\approx} f_{S_2}] \cdot (\zeta' n)
\leqs (\kappa + \zeta') n. \label{ineq:maj-upper}
\end{align}

We can then lower bound the expression on the left hand side as follows.
\begin{align}
\EX_{S_1, S_2 \in \cS'}[\disa(f_{S_1}, f_{S_2})] &= \sum_{x \in [n]} \Pr_{S_1, S_2 \in \cS'}\left[x \in S_1 \wedge x \in S_2 \wedge f_{S_1}(x) \ne f_{S_2}(x)\right] \nonumber \\
&\geqs \sum_{x \in [n]} \Pr_{S_1, S_2 \in \cS'}[x \in S_1 \wedge x \in S_2 \wedge f_{S_1}(x) \ne g(x) \wedge f_{S_2}(x) = g(x)] \nonumber \\
&= \sum_{x \in [n]} \Pr_{S_1 \in \cS'}[x \in S_1 \wedge f_{S_1}(x) \ne g(x)]\Pr_{S_2 \in \cS'}[x \in S_2 \wedge f_{S_2}(x) = g(x)] \nonumber \\
(\text{Since } g(x) = \maj_{S \in \cS' \atop x \in S} (f_{S}(x))) &\geqs \sum_{x \in [n]} \Pr_{S_1 \in \cS'}[x \in S_1 \wedge f_{S_1}(x) \ne g(x)]\Pr_{S_2 \in \cS'}[x \in S_2 \wedge f_{S_2}(x) \ne g(x)] \nonumber \\
&= \sum_{x \in [n]} \left(\Pr_{S \in \cS'}[x \in S \wedge f_S(x) \ne g(x)]\right)^2 \nonumber \\
(\text{Power Mean Inequality}) &\geqs \frac{1}{n} \left(\sum_{x \in [n]}\Pr_{S \in \cS'}[x \in S \wedge f_S(x) \ne g(x)]\right)^2 \nonumber \\
&= \frac{1}{n} \left(\EX_{S \in \cS'}\left[\disa(g, f_S)\right]\right)^2. \label{ineq:maj-lower}
\end{align}
Combining (\ref{ineq:maj-upper}) and (\ref{ineq:maj-lower}) gives the desired bound.
%We can then lower bound the expression on the left hand side as follows.
%\begin{align*}
%\EX_{S_1, S_2 \in \cS}[\disa(f_{S_1}, f_{S_2})] &= \EX_{S_1, S_2 \in \cS}\left[\sum_{x \in [n]} \ind[x \in S_1 \wedge x \in S_2 \wedge f_{S_1}(x) \ne f_{S_2}(x)]\right] \\
%&= \sum_{x \in [n]} \EX_{S_1, S_2 \in \cS}\left[\ind[x \in S_1 \wedge x \in S_2 \wedge f_{S_1}(x) \ne f_{S_2}(x)]\right] \\
%&= \sum_{x \in [n]} \Pr_{S_1, S_2 \in \cS}[x \in S_1 \wedge x \in S_2 \wedge f_{S_1}(x) \ne f_{S_2}(x)] \\
%&\geqs \sum_{x \in [n]} \Pr_{S_1, S_2 \in \cS}[x \in S_1 \wedge x \in S_2 \wedge f_{S_1}(x) \ne g(x) \wedge f_{S_2}(x) = g(x)] \\
%&= \sum_{x \in [n]} \Pr_{S_1 \in \cS}[x \in S_1 \wedge f_{S_1}(x) \ne g(x)]\Pr_{S_2 \in \cS}[x \in S_2 \wedge f_{S_2}(x) = g(x)] \\
%(\text{Since } g(x) = \maj_{S \in \cS' \atop x \in S} (f_{S}(x))) &\geqs \sum_{x \in [n]} \Pr_{S_1 \in \cS}[x \in S_1 \wedge f_{S_1}(x) \ne g(x)]\Pr_{S_2 \in \cS}[x \in S_2 \wedge f_{S_2}(x) \ne g(x)] \\
%&= \sum_{x \in [n]} \left(\Pr_{S \in \cS}[x \in S \wedge f_S(x) \ne g(x)]\right)^2 \\
%(\text{Power Mean Inequality}) &\geqs \frac{1}{n} \left(\sum_{x \in [n]}\Pr_{S \in \cS}[x \in S \wedge f_S(x) \ne g(x)]\right)^2 \\
%&= \frac{1}{n} \left(\sum_{x \in [n]}\EX_{S \in \cS}[\ind[x \in S \wedge f_S(x) \ne g(x)]]\right)^2 \\
%&= \frac{1}{n} \left(\EX_{S \in \cS}[\sum_{x \in [n]} \ind\left[x \in S \wedge f_S(x) \ne g(x)]\right]\right)^2 \\
%&= \frac{1}{n} \left(\EX_{S \in \cS}\left[\disa(g, f_S)\right]\right)^2. \\
%\end{align*}
\end{proof}

\subsection{Putting Things Together: Proof of Theorem~\ref{thm:agr-main-formal}} \label{sec:agr-proof}

Finally, we put all five steps together as outlined at the beginning of this section. This is formalized below. Note that Theorem~\ref{thm:agr-main-formal} follows from the theorem below simply by Markov inequality.

\begin{theorem} \label{thm:agr-main}
For any $0 < \eta, \zeta, \gamma, \mu < 1$ and $r, \ell, k, h, n, d \in \N$ such that $\ell \geqs 2$, let $\cS$ be any collection of $k$ subsets of $[n]$ such that $\cS$ is $(r, \ell, \zeta)$-intersection disperser and every subcollection $\tcS \subseteq \cS$ of size $h$ is $(\gamma, \mu)$-uniform, and let $\cF = \{f_S\}_{S \in \cS}$ be any collection of functions. If $\delta \triangleq \agr(\cF) \geqs \frac{10 + 64(r\ell)^2 k^{1/\ell}}{k}$, then there exists a subcollection $\cS' \subseteq \cS$ of size at least $\frac{\delta k}{128 \ell^2}$ and a function $g: [n] \to \{0, 1\}$ such that $$\EX_{S \in \cS'}\left[\disa(g, f_S)\right] \leqs n \sqrt{\frac{65536 h \ell^6}{\delta k} + \mu + 2 \zeta / \gamma}.$$
\end{theorem}

\begin{proof}
Observe that $\agr(\cF)$ directly corresponds to the number of blue edges $|E^{\cF, \zeta}_b|$ in the two-level consistency graph $G^{\cF, \zeta}$. In particular, $\agr(\cF) = \delta$ means that the number of blue edges is $(\delta k^2 - k) / 2$. Since $\cS$ is a $(r, \ell, \zeta)$-intersection disperser, Lemma~\ref{lem:transitivity} implies that $G^{\cF, \zeta}$ is $((r\ell)^{2(\ell - 1)}, \ell)$-red/blue-transitive. Let $d = \lfloor \frac{|E^{\cF, \zeta}|}{2k} \rfloor = \lfloor \frac{\delta k - 1}{4} \rfloor$; since $\delta \geqs \frac{10 + 64(r\ell)^2 k^{1/\ell}}{k}$, we have $d \geqs (r\ell)^2 k^{1/\ell}$.

Applying Lemma~\ref{lem:finding-dense-subgraph} with $G = G^{\cF, \zeta}, k_0 = k, \ell_0 = \ell, q_0 = (r\ell)^{2(\ell - 1)}$ and $d_0 = d$, we can conclude that there exist subsets $U_1, U_2 \subseteq V^{\cF, \zeta}$ each of size at least $d$ such that
\begin{align*}
\frac{|\{(u, v) \in U_1 \times U_2 \mid \{u, v\} \notin E^{\cF, \zeta}_r\}|}{|U_1||U_2|} \geqs \frac{1 - (r\ell)^{2(\ell - 1)}k/d^\ell}{\binom{\ell}{2}} \geqs \frac{1}{\ell^2}
\end{align*}
where the last inequality follows from our aforementioned lower bound on $d$ and from $\ell \geqs 2$.

Next, observe that, if we let $U'_1$ and $U'_2$ be random subsets of $U_1, U_2$ of size $d$, then we have
\begin{align*}
\EX_{U'_1, U'_2}\left[\frac{|\{(u', v') \in U'_1 \times U'_2 \mid \{u', v'\} \notin E^{\cF, \zeta}_r\}}{d^2}\right] = \frac{|\{(u, v) \in U_1 \times U_2 \mid \{u, v\} \notin E^{\cF, \zeta}_r\}|}{|U_1||U_2|}.
\end{align*}
As a result, there exists $\tU_1, \tU_2$ each of size exactly $d$ such that 
\begin{align}
\frac{|\{(\tu, \tv) \in \tU_1 \times \tU_2 \mid \{\tu, \tv\} \notin E^{\cF, \zeta}_r\}|}{d^2} \geqs \frac{1}{\ell^2}.
\label{eq:agr-1}
\end{align}

Now, let $\tU = \tU_1 \cup \tU_2$. (\ref{eq:agr-1}) implies that the number of $\{\tu, \tv\} \subseteq \tU$ such that $\{\tu, \tv\} \notin E^{\cF, \zeta}_r$ is at least $d^2/(2\ell^2) - d$ where the factor of 2 comes from the fact that each pair $\{\tu, \tv\}$ is double counted in the left hand side of (\ref{eq:agr-1}) and the subtraction of $d$ comes from the fact that the left hand side of (\ref{eq:agr-1}) also count the case where $\tu = \tv$.

Now, let $\tcF = \{f_S\}_{S \in \tU}$, $\zeta' = \mu + 2\zeta/\gamma$ and consider the two-level consistency graph $G^{\tcF, \zeta, \zeta'}$. Observe that $\{\tu, \tv\}$ is a blue edge in this new graph $G^{\tcF, \zeta, \zeta'}$ if and only if it is not a red edge in the original graph $G^{\cF, \zeta}$. Hence, the bound derived in the previous paragraph implies that
\begin{align*}
|E^{\tcF, \zeta, \zeta'}_b| \geqs \frac{d^2}{2\ell^2} - d.
\end{align*}
Let $d' = |E^{\tcF, \zeta, \zeta'}_b|/(2|\tU|) \geqs |E^{\tcF, \zeta, \zeta'}_b|/(4d) = d/(8\ell^2) - 1/4$. Recall that $d = \lfloor (\delta k - 1)/4 \rfloor$; from $\delta \geqs (10 + 64(r\ell)^2)/k$, we have $d \geqs 8\ell^2$ and $d \geqs \delta k / 8$. Hence, we have $d' \geqs d/(16\ell^2) \geqs \delta k / (128 \ell^2)$.

Furthermore, by Lemma~\ref{lem:transitivity-2} and from our assumption that every subcollection of $\cS$ of size $h$ is $(\gamma, \mu)$-uniform, the graph $G^{\cF,\zeta,\zeta'}$ is $h$-red/blue transitive. Applying Lemma~\ref{lem:finding-dense-subgraph} with $G = G^{\cF,\zeta,\zeta'}, k_0 = |\tU| \leqs 2d, \ell_0 = 2,  q_0 = h$ and $d_0 = d'$, there must be a set $U' \subseteq \tU$ of size at least $d'$ such that
\begin{align}
\frac{|\{(u', v') \in U' \times U' \mid \{u', v'\} \notin E^{\cF,\zeta,\zeta'}_r\}|}{|U'|^2} \geqs 1 - \frac{2hd}{(d')^2} \geqs 1 - \frac{512 h \ell^4}{d} \geqs 1 - \frac{65536 h \ell^6}{\delta k} 
\label{eq:agr-main-10}
\end{align}
where the last two inequalities follow from $d' \geqs d/(16\ell^2)$ and $d' \geqs \delta k/(128 \ell^2)$ respectively.

Let $\cF' = \{f_S\}_{S \in U'}$. Observe that the expression on the left hand side of (\ref{eq:agr-main-10}) is simply $\agr_{\zeta'}(\cF')$. Hence, by Lemma~\ref{lem:maj-decoding}, there exists a function $g: [n] \to \{0, 1\}$ such that 
\begin{align*}
\EX_{S \in U'}[\disa(g, f_S)] \leqs n\sqrt{\frac{65536 h \ell^6}{\delta k}  + \zeta'} = n\sqrt{\frac{65536 h \ell^6}{\delta k} + \mu + 2\zeta/\gamma}.
\end{align*}
In other words, $U'$ is the desired subcollection, which completes our proof.
\end{proof}

\section{Soundness Analysis of the Reduction}
\label{sec:soundness-csp}

We will next use our agreement theorem to analyze the soundness of our reduction. The soundness of our reduction can be stated more precisely as follows:

\begin{theorem} \label{thm:soundness-main}
For any $\Delta \in \N$, let $\Phi$ be any 3-CNF formula with variable set $\cX$ and clause set $\cC$ such that each variable appears in at most $\Delta$ clauses. Moreover, for any $0 < \eta, \zeta, \gamma, \mu < 1$ and $r, \ell, k, h, d \in \N$ such that $\ell \geqs 2$, let $\cT$ be any collection of $k$ subsets of $\cC$ such that $\cT$ is $(r, \ell, \zeta)$-intersection disperser and every subcollection $\tcT \subseteq \cT$ of size $h$ is $(\gamma, \mu)$-uniform. If $\val(\Phi) < 1 - \mu - (3\Delta/\gamma)\sqrt{4\Delta\mu + 6\Delta\zeta/\gamma}$, then $$\val(\Gamma_{\Phi, \cT}) < \frac{10 + 64(r\ell)^2 k^{1/\ell} + 65536h\ell^2/\mu}{k}.$$
\end{theorem}

Again, we will prove the contrapositive that if $\val(\Gamma_{\Phi, \cT})$ is large, then $\val(\Phi)$ is also large. Recall that $\val(\Gamma_{\Phi, \cT})$ being large implies that there exists a labeling $\sigma = \{\sigma_T\}_{T \in \cT}$ with high agreement probability. We would like to apply our agreement testing theorem. Note however that Theorem~\ref{thm:agr-main} only applies when the subsets of \emph{variables} are ``well-behaved'' (i.e. satisfies uniformity and is an intersection disperser). However, in our construction, the subset of variables are not random, rather they are variable set of random subsets of clauses. Hence, we will first need to translate the ``well-behavedness'' from subsets of clauses to their corresponding variable sets; this is shown in Section~\ref{subsec:well-behave-translation}. Once this is in place, we can apply Theorem~\ref{thm:agr-main}, which gives us a global assignment that approximately agrees with many $\sigma_T$'s. We show in Section~\ref{subsec:good-assignment} that such assignment satisfies most of the constraint, which implies that $\val(\Phi)$ must be large as desired. The full proof of Theorem~\ref{thm:soundness-main} can then be found in Section~\ref{sec:final}.

\subsection{Well-Behave Subsets of Clauses vs Well-Behave Subsets of Variables}
\label{subsec:well-behave-translation}

For convenient, let us define an additional notation:

\begin{definition}
Let $\Phi$ be any 3-CNF formula and $\cT$ be any subset of clauses of $\Phi$. We use $\cS_{\Phi, \cT}$ to denote the collection $\{\var(T)\}_{T \in \cT}$ of subsets of variables.
\end{definition}

Note that the subsets in $\cS_{\Phi, \cT}$ are indeed the variable sets of our labeling $\sigma = \{\sigma_T\}_{T \in \cT}$. Moreover, it is rather straightforward to see that both uniformity and intersection disperser conditions translate from $\cT$ to $\cS_{\Phi, \cT}$ with little loss in parameters, provided that each variable appears in bounded number of clauses. These observations are formalized and proved below.

\begin{lemma} \label{lem:uni-t}
Suppose that every variable in $\Phi$ appears in at least one and at most $\Delta$ clauses. If $\cT$ is $(\gamma, \mu)$-uniform, then $\cS_{\Phi, \cT}$ is $(\gamma, 3\Delta\mu)$-uniform.
\end{lemma}

\begin{proof}
First, observe that, since each variable appears in at most $\Delta$ clauses, we have $n \geqs m/\Delta$. Now, let $\cC_{\geqs \gamma} = \{C \in \cC \mid \Pr_{T \in \cT}[C \in T] \geqs \gamma\}$ and $\cX_{\geqs \gamma} = \{x \in \cX \mid \Pr_{S \in \cS_{\Phi, \cT}}[x \in S] \geqs \gamma\}$. Recall that $(\gamma, \mu)$-uniformity of $\cT$ implies that $|C_{\geqs \gamma}| \geqs (1 - \mu)m$. Observe that any $x \in \var(\cC_{\geqs \gamma})$ must also be contained in $\cX_{\geqs \gamma}$. Since every variable appears in at least one clauses, we have that every variable $x \notin \cX_{\geqs \gamma}$ must be in $\var(\cC \setminus \cC_{\geqs \gamma})$. As a result, $|\cX \setminus \cX_{\geqs \gamma}| \leqs 3\mu m$. From this and from $n \geqs m/\Delta$, we arrive at the desired conclusion.
\end{proof}

\begin{lemma} \label{lem:disp-t}
Suppose that every variable in $\Phi$ appears in at least one and at most $\Delta$ clauses. If $\cT$ is an $(r, \ell, \eta)$-intersection disperser, then $\cS_{\Phi, \cT}$ is $(r, \ell, 3\Delta\eta)$-uniform.
\end{lemma}

\begin{proof}
Consider any $r$ disjoint subcollections $\cS^1 = \{S_{1, 1}, \dots, S_{1, p_1}\}, \cdots, \cS^r = \{S_{r, 1}, \dots, S_{r, p_r}\} \subseteq \cS_{\Phi, \cT}$ each of size at most $\ell$. From our definition of $\cS_{\Phi, \cT}$, there is an $r$ disjoint subcollections $\cT^1 = \{T_{1, 1}, \dots, T_{1, p_1}\}, \dots, \cT^r = \{T_{r, 1}, \dots, T_{r, p_r}\} \subseteq \cT$ such that $S_{i, j} = \var(T_{i, j})$ for all $i \in [r]$ and $j \in [p_i]$. Observe that
\begin{align*}
\bigcup_{i=1}^{r}\left(\bigcap_{S \in \cS^{i}} S\right) \supseteq \var\left(\bigcup_{i=1}^{r}\left(\bigcap_{T \in \cT^{i}} T\right)\right).
\end{align*}
Moreover, since $\cT$ is an $(r, \ell, \eta)$-intersection disperser, we have $|\bigcup_{i=1}^{r}\left(\bigcap_{T \in \cT^{i}} T\right)| \geqs (1 - \eta)m$. As a result, since each variable appears in at least one clause, we indeed have $|\bigcup_{i=1}^{r}\left(\bigcap_{S \in \cS^{i}} S\right)| \geqs n - 3\eta m \geqs (1 - 3\Delta\eta)n$ as desired.
\end{proof}

\subsection{Global Function with Many Agreements is a Good Assignment}
\label{subsec:good-assignment}

In this subsection, we show that any global assignment that are approximately consistent with a collection of labels $\{\sigma_T\}_{T \in \cT^*}$ must satisfy most of the constraints, assuming that $\cT^*$ is sufficiently uniform, which is stated more precisely below. 

\begin{lemma} \label{lem:decoding}
Let $\cT^*$ be any $(\gamma, \mu)$-uniform collection of subsets of clauses and $\sigma$ be any labeling of $\cT^*$. If there exists $\psi: \cX \to \{0, 1\}$ such that $\EX_{T \in \cT^*}\left[\disa(\psi, \sigma_T)\right] \leqs \nu n$, then $\val(\psi) \geqs 1 - \mu - 3\ \nu \Delta / \gamma$.
\end{lemma}

The key to proving that $\psi$ violates few clauses is that, if a clause $C$ is violated, then, for each $T \in \cT^*$ that contains $T$, $\sigma_T$ and $\psi$ must disagree on at least one of $\var(C)$ because $\sigma_T$ satisfies $C$ but $\psi$ violates it. Hence, if $C$ appears often in $\cT$, then it contributes to many disagreements between $\sigma_T$ and $\psi$; the uniformity condition helps us ensure that most $C$ indeed appear often in $\cT$. Comparing this lower bound against the assumed upper bound on the expected disagreements gives us the desired result. This intuition is formalized below.

\begin{proof}
%Since $\val(\Phi) \geqs \val(\psi)$, it suffices to show that $\val(\psi) \geqs 1 - \mu m - 3\nu \Delta m / \gamma$.

Let $\cC_{\geqs \gamma}$ denote the set of all clauses that appear in at least $\gamma$ fraction of $T \in \cT^*$, i.e., $\cC_{\geqs \gamma} = \{C \in \cC \mid \Pr_{T \in \cT^*}[C \in T] \geqs \gamma\}$. Since $\cT^*$ is $(\gamma, \mu)$-uniform, we have $|\cC_{\geqs \gamma}| \geqs (1 - \mu)m$.

Since each variable $x$ appears in at most $\Delta$ clauses, we can obtain the following bound:
\begin{align}
\EX_{T \in \cT^*}[\disa(\psi, \sigma_T)] &= \sum_{x \in \cX} \Pr_{T \in \cT^*}[x \in \var(T) \wedge \sigma_T(x) \ne \psi(x)] \nonumber \\
&\geqs \frac{1}{\Delta} \sum_{C \in \cC_{\geqs \gamma}} \sum_{x \in \var(C)} \Pr_{T \in \cT^*}[x \in \var(T) \wedge \sigma_T(x) \ne \psi(x)] \nonumber \\
&\geqs \frac{1}{\Delta} \sum_{C \in \cC_{\geqs \gamma}} \sum_{x \in \var(C)} \Pr_{T \in \cT^*}[C \in T \wedge \sigma_T(x) \ne \psi(x)] \nonumber \\
(\text{Union Bound}) &\geqs \frac{1}{\Delta} \sum_{C \in \cC_{\geqs \gamma}} \Pr_{T \in \cT^*}\left[C \in T \wedge \left(\bigvee_{x \in \var(C)} \sigma_T(x) \ne \psi(x)\right)\right] \nonumber \\
&\geqs \frac{\gamma}{\Delta} \sum_{C \in \cC_{\geqs \gamma}} \Pr_{T \in \cT^*}\left[\bigvee_{x \in \var(C)} \sigma_T(x) \ne \psi(x) \text{ }\middle|\text{ } C \in T\right] \label{eq:biclique3}
\end{align}
Note here that we use the fact that each variable appears in at most $\Delta$ clauses in the first inequality and that the last inequality follows from the fact that each $C \in \cC_{\geqs \gamma}$ appears in at least $\gamma$ fraction of $T \in \cT^*$. The rest of the inequalities are trivial.

Let $\cC_{\unsat}$ denote the set of clauses violated by $\psi$. Observe that, for any $C \in \cC_{\unsat}$ and any $T \in \cT^*$ such that $C \in T$, $\sigma_T$ must disagree with $\psi$ on at least one of $x \in \var(C)$; this is simply because $C$ is satisfied by $\sigma_T$ but violated by $\psi$. In other words, for every $C \in C_{\unsat}$, we have
\begin{align}
\Pr_{T \in \cT^*}\left[\bigvee_{x \in \var(C)} \sigma_T(x) \ne \psi(x) \text{ }\middle|\text{ } C \in T\right] = 1. \label{eq:biclique4}
\end{align}

(\ref{eq:biclique3}), (\ref{eq:biclique4}) and the assumption that $\EX_{T \in \cT^*}[\disa(\psi, \sigma_T)] \leqs \nu n$ imply that
\begin{align*}
\nu \Delta n / \gamma \geqs |C_{\geqs \gamma} \cap \cC_{\unsat}|.
\end{align*}
Since $C_{\geqs \gamma} \geqs m(1 - \mu)$ and $n \leqs 3m$ (from every variable appears in at least one clause), we can conclude that $|\cC_{\unsat}| \leqs \mu m + 3\nu \Delta m / \gamma$. As a result, $\val(\psi) \geqs 1 - \mu - 3\nu \Delta / \gamma$ as desired.
\end{proof}

\subsection{Putting Things Together: Proof of Theorem~\ref{thm:soundness-main}}
\label{sec:final}

\begin{proof}[Proof of Theorem~\ref{thm:soundness-main}]
We may assume w.l.o.g. that each variable appears in at least one clause.

We will prove the theorem by contrapositive. Suppose that $\val(\Gamma_{\Phi, \cT}) \geqs \frac{10 + 64(r\ell)^2 k^{1/\ell} + 65536h\ell^2/\mu}{k}$. This means that there exists a labeling $\sigma = \{\sigma_T\}_{T \in \cT}$ such that $\val(\sigma) \geqs \frac{10 + 64(r\ell)^2 k^{1/\ell} + 65536h\ell^2/\mu}{k}$; this also means that, if we view $\sigma$ as a collection of functions $\cF = \{f_S\}_{S \in \cS_{\Phi, \cT}}$ where $f_{\var(T)} = \sigma_T$, then $\agr(\cF) \geqs \frac{10 + 64(r\ell)^2 k^{1/\ell} + 65536h\ell^2/\mu}{k}$. Let $\delta = \agr(\cF)$.

Furthermore, Lemmas~\ref{lem:disp-t} and~\ref{lem:uni-t} imply that $\cS_{\Phi, \cT}$ is an $(r, \ell, 3\Delta\zeta)$-intersection disperser and every subcollection of $\cS_{\Phi, \cT}$ of size $h$ is $(\gamma, 3\Delta\mu)$-uniform respectively. This enables us to apply Theorem~\ref{thm:agr-main} on $\cF$, which yields a subcollection $\cS' \subseteq \cS_{\Phi, \cT}$ of size at least $\frac{\delta k}{128 \ell^2} \geqs h$ and $g: \cX \to \{0, 1\}$ such that $$\EX_{S \in \cS'}\left[\disa(g, f_S)\right] \leqs n \sqrt{\frac{65536h \ell^6}{\delta k} + 3\Delta\mu + 6 \Delta \zeta / \gamma} \leqs n\sqrt{4\Delta\mu + 6 \Delta \zeta / \gamma}$$
where the second inequality follows from $\delta k \geqs 65536h\ell^2/\mu$.

Let $\cS^*$ be the subcollection of $\cS'$ of size $h$ that minimizes $\EX_{S \in \cS^*}\left[\disa(g, f_S)\right]$. Observe that $\EX_{S \in \cS^*}\left[\disa(g, f_S)\right] \leqs \EX_{S \in \cS'}\left[\disa(g, f_S)\right] \leqs n\sqrt{4\Delta\mu + 6 \Delta \zeta / \gamma}$. This is equivalent to saying that there exists a subcollection $\cT^* \subseteq \cT$ of size $h$ such that $\EX_{T \in \cT^*}\left[\disa(g, \sigma_T)\right] \leqs n\sqrt{4\Delta\mu + 6 \Delta \zeta / \gamma}$.

Since every subcollection of $\cT$ of size $h$ is $(\gamma, \mu)$-uniform, we can apply Lemma~\ref{lem:decoding} to infer that $\val(\Phi) \geqs 1 - \mu - (3\Delta/\gamma)\sqrt{4\Delta\mu + 6 \Delta \zeta / \gamma}$ as desired.
\end{proof}

\section{Proof of Inapproximability Results of 2-CSPs} \label{sec:2csp-hardness}

The inapproximability results for 2-CSPs can be shown simply by plugging in appropriate parameters to Theorem~\ref{thm:soundness-main}. More specifically, for ETH-hardness, since there is a $\polylog m$ loss in the PCP Theorem (Theorem~\ref{thm:d-pcp}), we need to select our $\alpha = 1/\polylog m$ so that the size (and running time) of the reduction is $2^{o(m)}$. Now, observe that the parameter $r$ in Theorem~\ref{thm:soundness-main} for the intersection disperser property grows with $(1/\alpha)^\ell$ (see Lemma~\ref{lem:well-behaved-set-deterministic}). Since the soundness guarantee in Theorem~\ref{thm:soundness-main} is of the form $k^{O(1/\ell)}(r\ell)^{O(1)}/k = k^{O(1/\ell)}(1/\alpha)^{O(\ell)}/k$, it is minimized when $\ell$ is roughly $\sqrt{\log k}$, which yields the soundness of $2^{(\log k)^{1/2 + o(1)}}/k$. Other parameters are chosen accordingly.

\begin{proof}[Proof of Theorem~\ref{thm:eth-hardness}]
Let $c, \varepsilon, \Delta$ be constants from Theorem~\ref{thm:d-pcp}.

For any 3-CNF formula $\tPhi$ with $m$ clauses, let us first apply the nearly-linear size PCP from Theorem~\ref{thm:d-pcp} to produce a 3-CNF formula $\Phi$ with $m' = O(m\log^c m)$ clauses. Let us also define the following parameters:
\begin{itemize}
\item $\alpha = \frac{1}{(\log m)^{c + 1}}$,
\item $\gamma = \frac{\alpha}{2} = \frac{1}{2(\log m)^{c + 1}}$,
\item $\mu = \frac{\varepsilon^2\gamma^2}{288\Delta^3} = \Theta_{\varepsilon, \Delta}\left(\frac{1}{(\log m)^{2c + 2}}\right)$,
\item $\zeta = \frac{\varepsilon^2\gamma^3}{432\Delta^3} = \Theta_{\varepsilon, \Delta}\left(\frac{1}{(\log m)^{3c + 3}}\right)$,
\item $\ell = (\log m)^{1/4}$,
\item $r = \lceil \frac{\ln(2/\zeta)}{\alpha^\ell} \rceil = 2^{\Theta_{\varepsilon, \Delta, c}((\log m)^{1/4} \log \log m)}$,
\item $h = \lceil 8\ln(2/\mu)/\alpha \rceil = \Theta_{\varepsilon, \Delta, c}((\log m)^{c + 1})$,
\item $k = 2^{\ell^2} = 2^{\sqrt{\log m}}$.
\end{itemize}
We then use Lemma~\ref{lem:well-behaved-set-deterministic} with the above parameters $\alpha, \mu, \zeta, k, \ell$ to construct a collection $\cS$ of subsets of clauses of $\Phi$ such that the following conditions hold.
\begin{itemize}
\item Every subset in $\cS$ has size at most $2\alpha m' = o(m)$.
\item $\cS$ is a $(r, \ell, \zeta)$-disperser.
\item Any subcollection $\tcS \subseteq \cS$ of size $h$ is $(\alpha/2, \mu)$-uniform.
\end{itemize}
Note that, for our choices of parameter, the parameter $m_0$ of Lemma~\ref{lem:well-behaved-set-deterministic} is $2^{O_{\varepsilon, \Delta, c}((\log m)^{1/4} \log \log m)} = 2^{o(\log m)}$. This means that, for sufficiently large $m$, we indeed have $m' \geqs m \geqs m_0$ and the running time needed to produce $\cS$ is $\poly(m) 2^{O(m_0 k^2)} = 2^{o(m)}$. Note that we assume without loss of generality here that $m' \geqs m$; if this is not the case, we can simply copy the formula $\Phi$ $\lceil m/m' \rceil$ times using new variables each time, which does not change the value of the formula.

We now consider the 2-CSP instance $\Gamma_{\Phi, \cS}$. Observe that the running time used to create $\Gamma_{\Phi, \cS}$ (and hence also the size of $\Gamma_{\Phi, \cS}$) is no more than $\poly(k) \cdot 2^{o(\alpha m')} = 2^{o(m)}$. Moreover, if $\val(\tPhi) = 1$, then $\val(\Phi) = 1$ and it is easy to see that $\val(\Gamma_{\Phi, \cS}) = 1$ as well.

On the other hand, if $\val(\tPhi) < 1$, then $\val(\Phi) < 1 - \varepsilon$. Due to our choice of parameters, we can apply Theorem~\ref{thm:soundness-main}, which implies that 
\begin{align*}
\val(\Gamma_{\Phi, \cS}) < \frac{O((r\ell)^2k^{1/\ell} + h\ell^2/\mu)}{k} = 2^{O_{\varepsilon, \Delta, c}((\log m)^{1/4} \log \log m)}/k = 2^{O_{\varepsilon, \Delta, c}(\sqrt{\log k}\log \log k)}/k.
\end{align*}
For sufficiently large $m$ (depending only on $c, \varepsilon, \Delta, \rho$), this term is at most $2^{(\log k)^{1/2 + \rho}}/k$. Hence, if there exists a polynomial time that can distinguish the two cases in the theorem statement, we can run this algorithm on $\Gamma_{\Phi, \cS}$ to decide whether $\tPhi$ is satisfiable in $2^{o(m)}$ time, contradicting ETH.
\end{proof}

For Gap-ETH-hardness, we do not incur a loss of $\polylog m$ from the PCP Theorem anymore. Thus, it suffices to chose $\alpha$ to be any function that converges to zero as $k$ goes to $\infty$ (e.g. $\alpha = 1/\loglog k$), and $k$ can now be independent of $m$. The rest of the analysis remains unchanged.

\begin{proof}[Proof of Theorem~\ref{thm:gap-eth-hardness}]
Let $\delta, \varepsilon, \Delta$ be the constants from Theorem~\ref{thm:gap-eth-hardness}. For any positive integer $k$, define the parameters as follows:
\begin{itemize}
\item $\alpha = \frac{1}{\log \log k}$,
\item $\gamma = \frac{\alpha}{2} = \frac{1}{2(\log \log k)}$,
\item $\mu = \frac{\varepsilon^2\gamma^2}{288\Delta^3} = \Theta_{\varepsilon, \Delta}\left(\frac{1}{(\log \log k)^2}\right)$,
\item $\zeta = \frac{\varepsilon^2\gamma^3}{432\Delta^3} = \Theta_{\varepsilon, \Delta}\left(\frac{1}{(\log \log k)^3}\right)$,
\item $\ell = \sqrt{\log k}$,
\item $r = \lceil \frac{\ln(2/\zeta)}{\alpha^\ell} \rceil = 2^{\Theta_{\varepsilon, \Delta, c}(\sqrt{\log k} \log \log \log k)}$,
\item $h = \lceil 8\ln(2/\mu)/\alpha \rceil = \Theta_{\varepsilon, \Delta, c}(\log \log k)$.
\end{itemize}

Consider any 3-CNF formula $\Phi$ with $m$ clauses such that each variable appears in at most $\Delta$ clauses. We then use Lemma~\ref{lem:well-behaved-set-deterministic} with the above parameters $\alpha, \mu, \eta, k, \ell$ to construct a collection $\cS$ of subsets of clauses of $\Phi$ such that the following conditions hold.
\begin{itemize}
\item Every subset in $\cS$ has size at most $2\alpha m$.
\item $\cS$ is a $(r, \ell, \eta)$-disperser.
\item Any subcollection $\tcS \subseteq \cS$ of size $h$ is $(\alpha/2, \mu)$-uniform.
\end{itemize}
Note that, for our choices of parameter, the parameter $m_0$ of Lemma~\ref{lem:well-behaved-set-deterministic} is a function of $k$. This means that, for sufficiently large $m$ (which depends on $k$), we indeed have $m \geqs m_0$ and the running time needed to produce $\cS$ is $\poly(m) 2^{O(m_0 k^2)} = O_k(\poly(m))$.

We now consider the 2-CSP instance $\Gamma_{\Phi, \cS}$. Observe that the running time used to create $\Gamma_{\Phi, \cS}$ (and hence also the size of $\Gamma_{\Phi, \cS}$) is no more than $\poly(k) \cdot 2^{O(\alpha m)} = 2^{O(m/\log \log k)}$. Moreover, if $\val(\Phi) = 1$, it is easy to see that $\val(\Gamma_{\Phi, \cS}) = 1$ as well.

Suppose that $\val(\Phi) < 1 - \varepsilon$. Due to our choice of parameters, we can apply Theorem~\ref{thm:soundness-main}, which implies that $$\val(\Gamma_{\Phi, \cS}) < \frac{O(k^{1/\ell}(r\ell)^2) + h\ell^2/\mu}{k} = 2^{O_{\varepsilon, \Delta}(\log \log k/\sqrt{\log k})}/k.$$ For sufficiently large $k$ (depending on $\varepsilon, \Delta, \rho$), this term is at most $2^{(\log k)^{1/2 + \rho}}/k$.

If there exists a $g(k) \cdot (nk)^D$-time algorithm that can distinguish the two cases in the theorem statement for some constant $D$, then pick sufficiently large $k$ such that the time needed to produce $\Gamma_{\Phi, \cS}$ is $O(2^{\delta m})$ and its size is at most $2^{\delta m/D}$, and that $\val(\Gamma_{\Phi, \cS}) < 2^{1/(\log k)^{1/2 + \rho}}/k$ whenever $\val(\Phi) < 1 - \varepsilon$. When we run this algorithm on $\Gamma_{\Phi, \cS}$ for such $k$, the algorithm can distinguish between $\val(\Phi) = 1$ and $\val(\Phi) < 1 - \varepsilon$ in $O(2^{\delta m})$ time, which contradicts Gap-ETH.
\end{proof}

\section{Inapproximability of Directed Steiner Network} \label{sec:dsn}

We now move on to prove hardness of approximation of DSN by simply plugging in the our main theorems to known reductions from 2-CSPs to DSN. The properties of the reduction are stated in the lemma below. Note that, while the reduction is attributed to Dodis and Khanna~\cite{DK99}, the lemma below is extracted from~\cite{CFM17} since, in~\cite{DK99}, the full description and its properties are left out due to space constraint.

\begin{lemma}[{\cite[Lemma 27]{CFM17}}] \label{lem:red-dsn}
There exists a polynomial time reduction that, given a 2-CSP instance\footnote{Lemma 27 of~\cite{CFM17} states this reduction in terms of Maximum Colored Subgraph Isomorphism. However, it is easy to see that the reduction also works with 2-CSPs as well.} $\Gamma$ with the constraint graph being a complete graph on $k$ variables, produces an edge-weighted directed graph $G = (V, E)$ and a set of demand pairs $\cD = \{(s_1, t_1), \dots, (s_{k'}, t_{k'})\}$ such that
\begin{itemize}
\item (Completeness) If $\val(\Gamma) = 1$, then there exists a subgraph $H$ of cost $1$ that satisfies all demands.
\item (Soundness) If $\val(\Gamma) < \gamma$, every subgraph satisfying all demand pairs has cost more than $\sqrt{2/\gamma}$.  
\item (Parameter Dependency) $k' = k^2 - k$.
\end{itemize}
\end{lemma}

Note that the exponent $1/4$ in the hardness of approximating DSN comes from two places: we lose a square factor in the parameter (i.e. $k' = \Theta(k^2)$) and another square factor in the gap.

\begin{proof}[Proof of Corollary~\ref{cor:eth-hardness-dsn}]
Suppose for the sake of contradiction that, for some constant $\rho' > 0$, there exists a polynomial time $2^{(\log k')^{1/2 + \rho'}}/(k')^{1/4}$-approximation algorithm where $k'$ is the number of demand pairs; let us call this algorithm $\bbA$. Moreover, let $\rho$ be any constant smaller than $\rho'$.

Given a 2-CSP instance $\Gamma$ with complete constraint graph on $k$ variables, we invoke Lemma~\ref{lem:red-dsn} to produce a DSN instance $(G, D)$ where $|D| = k' = k^2 - k$. From the completeness of the construction, we have that, if $\val(\Gamma) = 1$, then the optimum of $(G, D)$ is also $1$. On the other hand, if $\val(\Gamma) < 2^{(\log k)^{1/2 + \rho}}/k$, then the optimum of $(G, \cD)$ must be more than $\sqrt{2k/2^{(\log k)^{1/2 + \rho}}}$, which is at least $(k')^{1/4}/2^{(\log k')^{1/2 + \rho'}}$ when $k$ is sufficiently large. Hence, by running algorithm $\bbA$, we can distinguish these two cases of $\Gamma$ in polynomial time. From Theorem~\ref{thm:eth-hardness}, this contradicts ETH.
\end{proof}

\begin{proof}[Proof of Corollary~\ref{cor:gap-eth-hardness-dsn}]
Suppose for the sake of contradiction that, for some constant $\rho' > 0$ and for some function $g$, there exists a $g(k') \cdot (nk')^{O(1)}$-time $2^{(\log k')^{1/2 + \rho'}}/(k')^{1/4}$-approximation algorithm where $k'$ is the number of demand pairs; let us call this algorithm $\bbA$. Moreover, let $\rho$ be any constant smaller than $\rho'$.

Given a 2-CSP instance $\Gamma$ with complete constraint graph on $k$ variables, we invoke Lemma~\ref{lem:red-dsn} to produce a DSN instance $(G, D)$ where $|D| = k' = k^2 - k$. From the completeness of the construction, if $\val(\Gamma) = 1$, then the optimum of $(G, D)$ is also $1$. On the other hand, if $\val(\Gamma) < 2^{(\log k)^{1/2 + \rho}}/k$, then the optimum of $(G, \cD)$ must be more than $\sqrt{2k/2^{(\log k)^{1/2 + \rho}}}$, which is at least $(k')^{1/4}/2^{(\log k')^{1/2 + \rho'}}$ when $k$ is sufficiently large. Hence, by running algorithm $\bbA$, we can distinguish these two cases of $\Gamma$ in time $t(k) \cdot |\Gamma|^{O(1)}$ where $t(k) = g(k^2 - k)$. From Theorem~\ref{thm:gap-eth-hardness}, this contradicts Gap-ETH.
\end{proof}

\section{Conclusion and Discussions} \label{sec:conclusion}

In this work, we show that 2-CSP is ETH-hard to approximate to within a factor of $k^{1 - o(1)}$ where $k$ denotes the number of variables. This ratio is nearly optimal since a trivial algorithm yields an $O(k)$-approximation for the problem. Under Gap-ETH, we strengthen our result by improving the lower order term in the inapproximability factor and ruling out not only polynomial time algorithm but FPT algorithms parameterized by $k$. Due to a known reduction, our results also imply $k^{1/4 - o(1)}$ hardness of approximating DSN where $k$ denotes the number of demand pairs. 

Of course the polynomial sliding scale conjecture still remains open after our work and, as touched upon in the introduction, resolving the conjecture will help advance our understanding of approximability of many problems. Even without fully resolving the conjecture, it may still be good to further study the interaction between the number of variables $k$ and the alphabet size $n$. For instance, while we show the inapproximability result with ratio almost linear in $k$, the dependency between $n$ and $k$ is quite bad; in particular, in our ETH-hardness reduction, $n$ is $2^{2^{(\log k)^d}}$ for some constant $d > 0$. Would it be possible to improve this dependency (say, to $n = k^{\polylog k}$)?

On the other hand, as explained earlier, $k$ must be independent of $n$ in the parameterized setting and hence our question above does not apply to this regime. However, one intriguing question in this area is whether a parameterized hardness of approximation for 2-CSPs can be proved under any assumption weaker than Gap-ETH (e.g. ETH). This is not known even for a constant inapproximability factor. In fact, it is not hard to see that inapproximability of parameterized 2-CSPs implies inapproximability of parameterized clique; the latter is a well-studied problem that has so far resisted attempts at proving inapproximability from any assumption except Gap-ETH~\cite{CHK13,HKK13,KS16}. Note that this is in contrast to some other parameterized problems, such as dominating set, for which $\W[1]$-hardness of approximation is known~\cite{CL16,CLM17}\footnote{See also~\cite{Lin15} which~\cite{CL16} relies heavily on.}.

Another interesting research direction is to try to prove similar hardness results for other problems. For example, Densest $k$-Subgraph (D$k$S) is one such candidate problem; similar to 2-CSPs with $k$ variables, the problem can be approximated trivially to within $O(k)$-factor and no polynomial time (or even FPT time parameterized by $k$) $k^{1 - \varepsilon}$-approximation algorithm is known for the problem for any $\varepsilon > 0$. Hence, it may be possible to prove ETH-hardness of factor $k^{1 - o(1)}$ for D$k$S as well.

\subsection*{Acknowledgment}

We would like to thank Prahladh Harsha for many useful discussions. Pasin would also like to thank Rajesh Chitnis and Andreas Emil Feldmann for insightful discussions regarding directed steiner network problems.

\bibliographystyle{alpha}
\bibliography{main}

\newcommand{\etalchar}[1]{$^{#1}$}
\begin{thebibliography}{CMMV17}

\bibitem[AB17]{AB17}
Amir Abboud and Greg Bodwin.
\newblock Reachability preservers: New extremal bounds and approximation
  algorithms.
\newblock {\em CoRR}, abs/1710.11250, 2017.

\bibitem[AIM14]{AIM14}
Scott Aaronson, Russell Impagliazzo, and Dana Moshkovitz.
\newblock {AM} with multiple merlins.
\newblock In {\em CCC}, pages 44--55, 2014.

\bibitem[ALM{\etalchar{+}}98]{ALMSS}
Sanjeev Arora, Carsten Lund, Rajeev Motwani, Madhu Sudan, and Mario Szegedy.
\newblock Proof verification and the hardness of approximation problems.
\newblock {\em J. ACM}, 45(3):501--555, May 1998.

\bibitem[AM09]{AM}
Per Austrin and Elchanan Mossel.
\newblock Approximation resistant predicates from pairwise independence.
\newblock {\em Computational Complexity}, 18(2):249--271, 2009.

\bibitem[App17]{App17}
Benny Applebaum.
\newblock Exponentially-hard gap-{CSP} and local {PRG} via local hardcore
  functions.
\newblock In {\em FOCS}, pages 836--847, 2017.

\bibitem[ARW17]{AR17}
Amir Abboud, Aviad Rubinstein, and Ryan Williams.
\newblock Distributed {PCP} theorems for hardness of approximation in {P}.
\newblock In {\em FOCS}, pages 25--36, 2017.

\bibitem[AS97]{ArSu}
Sanjeev Arora and Madhu Sudan.
\newblock Improved low degree testing and its applications.
\newblock In {\em STOC}, pages 485--495, May 1997.

\bibitem[AS98]{AS}
Sanjeev Arora and Shmuel Safra.
\newblock Probabilistic checking of proofs: A new characterization of {NP}.
\newblock {\em J. ACM}, 45(1):70--122, January 1998.

\bibitem[BBM{\etalchar{+}}13]{BBMRY13}
Piotr Berman, Arnab Bhattacharyya, Konstantin Makarychev, Sofya Raskhodnikova,
  and Grigory Yaroslavtsev.
\newblock Approximation algorithms for spanner problems and directed steiner
  forest.
\newblock {\em Inf. Comput.}, 222:93--107, 2013.

\bibitem[BEKP15]{BEKP15}
Edouard Bonnet, Bruno Escoffier, Eun~Jung Kim, and Vangelis~Th. Paschos.
\newblock On subexponential and {FPT}-time inapproximability.
\newblock {\em Algorithmica}, 71(3):541--565, 2015.

\bibitem[BGH{\etalchar{+}}06]{BGHSV06}
Eli Ben{-}Sasson, Oded Goldreich, Prahladh Harsha, Madhu Sudan, and Salil~P.
  Vadhan.
\newblock Robust {PCP}s of proximity, shorter {PCP}s, and applications to
  coding.
\newblock {\em {SIAM} J. Comput.}, 36(4):889--974, 2006.

\bibitem[BGKM18]{BGKM18}
Arnab Bhattacharyya, Suprovat Ghoshal, {Karthik {C. S.}}, and Pasin Manurangsi.
\newblock Parameterized intractability of even set and shortest vector problem
  from {Gap-ETH}.
\newblock In {\em ICALP}, 2018.
\newblock To appear.

\bibitem[BGLR93]{BGLR93}
Mihir Bellare, Shafi Goldwasser, Carsten Lund, and Alexander Russell.
\newblock Efficient probabilistically checkable proofs and applications to
  approximations.
\newblock In {\em STOC}, pages 294--304, 1993.

\bibitem[BKK{\etalchar{+}}16]{BKKMS16}
Eli Ben{-}Sasson, Yohay Kaplan, Swastik Kopparty, Or~Meir, and Henning
  Stichtenoth.
\newblock Constant rate {PCP}s for {Circuit-SAT} with sublinear query
  complexity.
\newblock {\em J. {ACM}}, 63(4):32:1--32:57, 2016.

\bibitem[BKRW17]{BKRW17}
Mark Braverman, Young Kun{-}Ko, Aviad Rubinstein, and Omri Weinstein.
\newblock {ETH} hardness for densest-\emph{k}-subgraph with perfect
  completeness.
\newblock In {\em SODA}, pages 1326--1341, 2017.

\bibitem[BKW15]{BKW15}
Mark Braverman, Young Kun{-}Ko, and Omri Weinstein.
\newblock Approximating the best nash equilibrium in $n^{o(\log n)}$-time
  breaks the exponential time hypothesis.
\newblock In {\em SODA}, pages 970--982, 2015.

\bibitem[BS08]{BS08}
Eli Ben{-}Sasson and Madhu Sudan.
\newblock Short {PCP}s with polylog query complexity.
\newblock {\em {SIAM} J. Comput.}, 38(2):551--607, 2008.

\bibitem[CCC{\etalchar{+}}99]{CCCDGGL99}
Moses Charikar, Chandra Chekuri, To{-}Yat Cheung, Zuo Dai, Ashish Goel, Sudipto
  Guha, and Ming Li.
\newblock Approximation algorithms for directed steiner problems.
\newblock {\em J. Algorithms}, 33(1):73--91, 1999.

\bibitem[CCK{\etalchar{+}}17]{CCKLMNT17}
Parinya Chalermsook, Marek Cygan, Guy Kortsarz, Bundit Laekhanukit, Pasin
  Manurangsi, Danupon Nanongkai, and Luca Trevisan.
\newblock From {Gap-ETH} to {FPT}-inapproximability: Clique, dominating set,
  and more.
\newblock In {\em FOCS}, pages 743--754, 2017.

\bibitem[CDKL17]{CDKL17}
Eden Chlamt{\'{a}}c, Michael Dinitz, Guy Kortsarz, and Bundit Laekhanukit.
\newblock Approximating spanners and directed steiner forest: Upper and lower
  bounds.
\newblock In {\em SODA}, pages 534--553, 2017.

\bibitem[CEGS11]{CEGS11}
Chandra Chekuri, Guy Even, Anupam Gupta, and Danny Segev.
\newblock Set connectivity problems in undirected graphs and the directed
  steiner network problem.
\newblock {\em {ACM} Trans. Algorithms}, 7(2):18:1--18:17, 2011.

\bibitem[CFM17]{CFM17}
Rajesh Chitnis, Andreas~Emil Feldmann, and Pasin Manurangsi.
\newblock Parameterized approximation algorithms for directed steiner network
  problems.
\newblock {\em CoRR}, abs/1707.06499, 2017.

\bibitem[Cha16]{Chan}
Siu~On Chan.
\newblock Approximation resistance from pairwise-independent subgroups.
\newblock {\em J. {ACM}}, 63(3):27:1--27:32, 2016.

\bibitem[Che18]{Chen18}
Lijie Chen.
\newblock On the hardness of approximate and exact (bichromatic) maximum inner
  product.
\newblock In {\em CCC}, 2018.
\newblock To appear.

\bibitem[CHK11]{CHK11}
Moses Charikar, MohammadTaghi Hajiaghayi, and Howard~J. Karloff.
\newblock Improved approximation algorithms for label cover problems.
\newblock {\em Algorithmica}, 61(1):190--206, 2011.

\bibitem[CHK13]{CHK13}
Rajesh~Hemant Chitnis, MohammadTaghi Hajiaghayi, and Guy Kortsarz.
\newblock Fixed-parameter and approximation algorithms: {A} new look.
\newblock In {\em IPEC}, pages 110--122, 2013.

\bibitem[CL16]{CL16}
Yijia Chen and Bingkai Lin.
\newblock The constant inapproximability of the parameterized dominating set
  problem.
\newblock In {\em FOCS}, pages 505--514, 2016.

\bibitem[CMMV17]{CMMV17}
Eden Chlamt{\'{a}}c, Pasin Manurangsi, Dana Moshkovitz, and Aravindan
  Vijayaraghavan.
\newblock Approximation algorithms for label cover and the log-density
  threshold.
\newblock In {\em SODA}, pages 900--919, 2017.

\bibitem[DFK{\etalchar{+}}11]{DFKRS11}
Irit Dinur, Eldar Fischer, Guy Kindler, Ran Raz, and Shmuel Safra.
\newblock {PCP} characterizations of {NP}: Toward a polynomially-small
  error-probability.
\newblock {\em Computational Complexity}, 20(3):413--504, 2011.

\bibitem[DFS16]{DFS16}
Argyrios Deligkas, John Fearnley, and Rahul Savani.
\newblock Inapproximability results for approximate nash equilibria.
\newblock In {\em WINE}, pages 29--43, 2016.

\bibitem[DHK15]{DHK15}
Irit Dinur, Prahladh Harsha, and Guy Kindler.
\newblock Polynomially low error {PCP}s with polyloglog n queries via modular
  composition.
\newblock In {\em STOC}, pages 267--276, 2015.

\bibitem[Din07]{D07}
Irit Dinur.
\newblock The {PCP} theorem by gap amplification.
\newblock {\em J. {ACM}}, 54(3):12, 2007.

\bibitem[Din16]{D16}
Irit Dinur.
\newblock Mildly exponential reduction from gap {3SAT} to polynomial-gap
  label-cover.
\newblock {\em ECCC}, 23:128, 2016.

\bibitem[DK99]{DK99}
Yevgeniy Dodis and Sanjeev Khanna.
\newblock Design networks with bounded pairwise distance.
\newblock In {\em STOC}, pages 750--759, 1999.

\bibitem[DK17]{DK17}
Irit Dinur and Tali Kaufman.
\newblock High dimensional expanders imply agreement expanders.
\newblock In {\em FOCS}, pages 974--985, 2017.

\bibitem[DM18]{DM18}
Irit Dinur and Pasin Manurangsi.
\newblock {ETH}-hardness of approximating 2-{CSP}s and directed steiner
  network.
\newblock In {\em ITCS}, pages 36:1--36:20, 2018.

\bibitem[DN17]{DN17}
Irit Dinur and Inbal~Livni Navon.
\newblock Exponentially small soundness for the direct product z-test.
\newblock In {\em CCC}, pages 29:1--29:50, 2017.

\bibitem[DS14]{DS14}
Irit Dinur and David Steurer.
\newblock Analytical approach to parallel repetition.
\newblock In {\em STOC}, pages 624--633, 2014.

\bibitem[FKN12]{FKN12}
Moran Feldman, Guy Kortsarz, and Zeev Nutov.
\newblock Improved approximation algorithms for directed steiner forest.
\newblock {\em J. Comput. Syst. Sci.}, 78(1):279--292, 2012.

\bibitem[Gri01]{Gri01}
Dima Grigoriev.
\newblock Linear lower bound on degrees of positivstellensatz calculus proofs
  for the parity.
\newblock {\em Theor. Comput. Sci.}, 259(1-2):613--622, 2001.

\bibitem[H{\aa}s01]{Hastad}
Johan H{\aa}stad.
\newblock Some optimal inapproximability results.
\newblock {\em J. {ACM}}, 48(4):798--859, 2001.

\bibitem[HKK13]{HKK13}
Mohammad~Taghi Hajiaghayi, Rohit Khandekar, and Guy Kortsarz.
\newblock The foundations of fixed parameter inapproximability.
\newblock {\em CoRR}, abs/1310.2711, 2013.

\bibitem[IKW12]{ImpagliazzoKW12}
Russell Impagliazzo, Valentine Kabanets, and Avi Wigderson.
\newblock New direct-product testers and 2-query {PCP}s.
\newblock {\em SIAM J. Comput.}, 41(6):1722--1768, 2012.

\bibitem[IP01]{IP01}
Russell Impagliazzo and Ramamohan Paturi.
\newblock On the complexity of k-{SAT}.
\newblock {\em J. Comput. Syst. Sci.}, 62(2):367--375, 2001.

\bibitem[IPZ01]{IPZ01}
Russell Impagliazzo, Ramamohan Paturi, and Francis Zane.
\newblock Which problems have strongly exponential complexity?
\newblock {\em J. Comput. Syst. Sci.}, 63(4):512--530, 2001.

\bibitem[Kho02]{Khot-ugc}
Subhash Khot.
\newblock On the power of unique 2-prover 1-round games.
\newblock In {\em STOC}, pages 767--775, 2002.

\bibitem[KKMO07]{KKMO}
Subhash Khot, Guy Kindler, Elchanan Mossel, and Ryan O'Donnell.
\newblock Optimal inapproximability results for {MAX-CUT} and other 2-variable
  {CSP}s?
\newblock {\em {SIAM} J. Comput.}, 37(1):319--357, 2007.

\bibitem[KKT16]{KKT16}
Guy Kindler, Alexandra Kolla, and Luca Trevisan.
\newblock Approximation of non-boolean {2CSP}.
\newblock In {\em SODA}, pages 1705--1714, 2016.

\bibitem[KLM18]{CLM17}
{Karthik {C. S.}}, Bundit Laekhanukit, and Pasin Manurangsi.
\newblock On the parameterized complexity of approximating dominating set.
\newblock In {\em STOC}, 2018.
\newblock To appear.

\bibitem[KS16]{KS16}
Subhash Khot and Igor Shinkar.
\newblock On hardness of approximating the parameterized clique problem.
\newblock In {\em ITCS}, pages 37--45, 2016.

\bibitem[KST54]{KST54}
Tam\'{a}s K\H{o}v\'{a}ri, Vera~T. S\'{o}s, and P\'{a}l Tur\'{a}n.
\newblock {On a problem of K. Zarankiewicz}.
\newblock {\em Colloquium Mathematicae}, 3(1):50--57, 1954.

\bibitem[Lin15]{Lin15}
Bingkai Lin.
\newblock The parameterized complexity of \emph{k}-{Biclique}.
\newblock In {\em SODA}, pages 605--615, 2015.

\bibitem[Man17]{M17}
Pasin Manurangsi.
\newblock Almost-polynomial ratio {ETH}-hardness of approximating densest
  $k$-subgraph.
\newblock In {\em STOC}, pages 954--961, 2017.

\bibitem[MM17]{MM17}
Pasin Manurangsi and Dana Moshkovitz.
\newblock Improved approximation algorithms for projection games.
\newblock {\em Algorithmica}, 77(2):555--594, 2017.

\bibitem[Mos17]{Mos17}
Dana Moshkovitz.
\newblock Low-degree test with polynomially small error.
\newblock {\em Computational Complexity}, 26(3):531--582, 2017.

\bibitem[MR10]{MR10}
Dana Moshkovitz and Ran Raz.
\newblock Sub-constant error probabilistically checkable proof of almost-linear
  size.
\newblock {\em Computational Complexity}, 19(3):367--422, 2010.

\bibitem[MR16]{MR16}
Pasin Manurangsi and Prasad Raghavendra.
\newblock A birthday repetition theorem and complexity of approximating dense
  {CSP}s.
\newblock {\em CoRR}, abs/1607.02986, 2016.

\bibitem[MR17]{MR17}
Pasin Manurangsi and Aviad Rubinstein.
\newblock Inapproximability of {VC} dimension and littlestone's dimension.
\newblock In {\em COLT}, pages 1432--1460, 2017.

\bibitem[Pel07]{Peleg07}
David Peleg.
\newblock Approximation algorithms for the label-cover\({}_{\mbox{max}}\) and
  red-blue set cover problems.
\newblock {\em J. Discrete Algorithms}, 5(1):55--64, 2007.

\bibitem[Rag08]{Rag08}
Prasad Raghavendra.
\newblock Optimal algorithms and inapproximability results for every {CSP}?
\newblock In {\em STOC}, pages 245--254, 2008.

\bibitem[Raz98]{Raz98}
Ran Raz.
\newblock A parallel repetition theorem.
\newblock {\em {SIAM} J. Comput.}, 27(3):763--803, 1998.

\bibitem[RS96]{RuSu96}
Ronitt Rubinfeld and Madhu Sudan.
\newblock Robust characterizations of polynomials with applications to program
  testing.
\newblock {\em {SIAM} J. Comput.}, 25(2):252--271, 1996.

\bibitem[RS97]{RazS97}
Ran Raz and Shmuel Safra.
\newblock A sub-constant error-probability low-degree test, and a sub-constant
  error-probability {PCP} characterization of {NP}.
\newblock In {\em STOC}, pages 475--484, 1997.

\bibitem[Rub16a]{Rub-itcs}
Aviad Rubinstein.
\newblock Detecting communities is hard, and counting them is even harder.
\newblock {\em CoRR}, abs/1611.08326, 2016.

\bibitem[Rub16b]{Rub16-focs}
Aviad Rubinstein.
\newblock Settling the complexity of computing approximate two-player nash
  equilibria.
\newblock In {\em FOCS}, pages 258--265, 2016.

\bibitem[Rub17]{Rub17}
Aviad Rubinstein.
\newblock Honest signaling in zero-sum games is hard, and lying is even harder.
\newblock In {\em ICALP}, pages 77:1--77:13, 2017.

\bibitem[Rub18]{Rub18}
Aviad Rubinstein.
\newblock Hardness of approximate nearest neighbor search.
\newblock In {\em STOC}, 2018.
\newblock To appear.

\bibitem[Sch08]{Sch08}
Grant Schoenebeck.
\newblock Linear level lasserre lower bounds for certain k-csps.
\newblock In {\em FOCS}, pages 593--602, 2008.

\end{thebibliography}

\appendix

\section{Constructing Well-Behaved Sets} \label{app:well-behaved-sets}

The main goal of this section is to prove Lemma~\ref{lem:well-behaved-set-deterministic}. We divide this section into two parts. In the first part (\ref{app:random-well-behaved}), we use probabilistic argument to show that a collection of random subsets satisfied properties as required in Lemma~\ref{lem:well-behaved-set-deterministic}. Then, in the section part (\ref{app:derandomize}), we show how to construct these subsets deterministically in subexponential time, which is required if a deterministic reduction from 3-SAT to 2-CSP is sought.

\subsection{Random Sets Behave Well} \label{app:random-well-behaved}

For any universe $\cU$ and any $0 < \alpha < 1$, let $\cD_{\cU, \alpha}$ denote the distribution on subsets of $\cU$ where $S \sim \cD_{\cU, \alpha}$ is generated by including each element $u \in \cU$ independently with probability $\alpha$.

The main lemma of this section is the following, which implies the existential part of Lemma~\ref{lem:well-behaved-set-deterministic}.

\begin{lemma} \label{lem:random-well-behaved}
For any $0 < \alpha, \mu, \eta < 1$ and any $m, k, \ell \in \N$, let $S_1, \dots, S_k$ be subsets of an $m$-element universe $\cU$ drawn independently at random from $\cD_{\cU, \alpha}$. Then, with probability at least $1 - 2^{\log k (\lceil8\ln(2/\mu)/\alpha\rceil) - \mu^2 m / 16} - 2^{\ell \lceil \ln(2/\eta)/\alpha^\ell\rceil \log(2k) - \eta m / 6} - 2^{-\alpha m / 3}$, the following properties hold:
\begin{itemize}
\item (Size) Every subset has size at most $2\alpha m$.
\item (Intersection Disperser) the collection $\cS = \{S_1, \dots, S_k\}$ is a $(\lceil \ln(2/\eta)/\alpha^\ell\rceil, \ell, \eta)$-disperser.
\item (Uniformity) Any subcollection $\tcS \subseteq \cS$ of size $\lceil 8\ln(2/\mu)/\alpha \rceil$ is $(\alpha/2, \mu)$-uniform.
\end{itemize}
\end{lemma}

By union bound, the above lemma is an immediate implication of Lemmas~\ref{lem:random-uniform},~\ref{lem:random-disperser} and~\ref{lem:random-size}, which are stated and proved below. All proofs consist only of straightforward probabilistic arguments.

\begin{lemma} \label{lem:random-uniform}
For any $0 < \alpha, \mu < 1$ and any $m, k \in \N$, let $S_1, \dots, S_k$ be subsets of an $m$-element universe $\cU$ drawn independently at random from $\cD_{\cU, \alpha}$. Then, for every subcollection $\tcS \subseteq \cS$ of size $\lceil8\ln(2/\mu)/\alpha\rceil$, $\tcS$ is $(\alpha/2, \mu)$-uniform with probability at least $1 - 2^{\log k (\lceil8\ln(2/\mu)/\alpha\rceil) - \mu^2 m / 16}$.
\end{lemma}

\begin{proof}
Consider any subcollection $\tcS = \{S_{j_1}, \dots, S_{j_h}\}$ where $h = \lceil 8\ln(2/\mu)/\alpha \rceil$. We will calculate the probability that $\tcS$ is $(\alpha/2, \mu)$-uniform and use union bound over all such $\tcS$'s to derived the desired result.

Recall that, for each $u \in \cU$, $u$ is included independently in each set $S_{j_i}$ with probability $\alpha$. Hence, we can apply Chernoff bound to lower bound the probability that $u$ is included in less than $\alpha/2$ fraction of the subsets in $\tcS$; in particular, this yields the following inequality.
\begin{align} \label{eq:ind1}
\Pr\left[|\{S \in \tcS \mid u \in S\}| \geqs (\alpha/2)|\tcS| \right] \geqs 1 - e^{-\alpha |\tcS|/8} \geqs 1 - \mu/2
\end{align}
where the second inequality comes from $|\tcS| = h \geqs 8\ln(2/\mu)/\alpha$.

Now, note again that the event $|\{S \in \tcS \mid u \in S\}| \geqs (\alpha/2)|\tcS|$ is independent for each $u \in \cU$. Hence, we can again apply Chernoff bound to lower bound the probability that this event occurs for at least $(1 - \mu)$ fraction of $u \in \cU$, which gives the following bound:
\begin{align*}
\Pr\left[|\{u \in \cU \mid \{S \in \tcS \mid u \in S\}| \geqs (\alpha/2)|\tcS|\}| \geqs (1 - \mu) m\right] \geqs 1 - e^{(-\mu^2(1 - \mu/2)m)/8} \geqs 1 - e^{-\mu^2m/16}.
\end{align*}
In other words, for each subcollection $\cS$ of size at least $8\ln(2/\mu)/\alpha$, $\cS$ is not $(\alpha/2, \mu)$-uniform with probability at most $e^{-\mu^2m/16}$. Since there are no more than $k^h$ such subcollection, union bound implies that the probability that every subcollection of size $h$ is $(\alpha/2, \mu)$-uniform is at least $1 - k^h e^{-\mu^2 m / 16} \geqs 1 - 2^{h \log k - \mu^2 m / 16}$.
\end{proof}

\begin{lemma} \label{lem:random-disperser}
For any $0 < \alpha, \eta < 1$ and any $m, k, \ell \in \N$, let $r = \lceil \ln(2/\eta) / \alpha^\ell \rceil$ and let $S_1, \dots, S_k$ be subsets of an $m$-element universe $\cU$ drawn independently at random from $\cD_{\cU, \alpha}$. Then, the collection $\cS = \{S_1, \dots, S_k\}$ is $(r, \ell, \eta)$-intersection disperser with probability at least $1 - 2^{\ell r \log(2k) - \eta m / 6}$.
\end{lemma}

\begin{proof}
Consider any disjoint subcollections $\cS^1 = \{S_{j_{1, 1}}, \dots, S_{j_{1, h_1}}\}, \dots, \cS^r = \{S_{j_{r, 1}}, \dots, S_{j_{r, h_r}}\}$ where $h_1, \dots, h_r \leqs \ell$. We will compute the probability that $\left|\bigcup_{i=1}^{r}\left(\bigcap_{S \in \cS^{i}} S\right)\right| \geqs (1 - \eta) m$ and then use union bound over all choices of $\cS^1, \dots, \cS^r$.

Let us consider an element $u \in \cU$. For each subcollection $\cS^i$, since $u$ is included in each of $S_{j_{i, 1}}, \dots, S_{j_{i, h_i}}$ independently with probability $\alpha$, $\Pr\left[u \in \bigcap_{S \in \cS^i} S\right] = \alpha^{h_i} \geqs \alpha^\ell$. Since the subcollections are disjoint, the event $u \notin \bigcap_{S \in \cS^i} S$ is independent for different $\cS^i$. Hence, we have
\begin{align} \label{eq:ind2}
\Pr\left[u \notin \bigcup_{i=1}^{r}\left(\bigcap_{S \in \cS^{i}} S\right)\right] = \prod_{i=1}^r \Pr\left[u \notin \bigcap_{S \in \cS^i} S\right] \leqs (1 - \alpha^\ell)^r \leqs e^{- \alpha^\ell r} \leqs \eta/2
\end{align}
where the last inequality comes from $r \geqs \ln(2/\eta)/\alpha^\ell$.

For different $u \in \cU$, the event $u \notin \bigcup_{i=1}^{r}\left(\bigcap_{S \in \cS^{i}} S\right)$ is independent. Applying the Chernoff bound, we have
\begin{align*}
\Pr\left[\left|\bigcup_{i=1}^{r}\left(\bigcap_{S \in \cS^{i}} S\right)\right| < (1 - \eta) m \right] \leqs e^{-\eta m / 6}.
\end{align*}

Finally, note that the number of different subcollections $\cS_1, \dots, \cS_r$ (i.e. $\{j_{1, 1}, \dots, j_{1, h_1}\}, \dots, \{j_{r, 1}, \dots, j_{r, h_r}\}$) is at most $(2k)^{\ell r}$; this is because there are at most $\binom{k}{0} + \binom{k}{1} + \cdots + \binom{k}{\ell} \leqs (\ell + 1)k^{\ell} \leqs (2k)^\ell$ choices of $\cS_i$ for each $i \in [r]$. As a result, by union bound, the probability that $\left|\bigcup_{i=1}^{r}\left(\bigcap_{S \in \cS^{i}} S\right)\right| < (1 - \eta) m$ is at most $(2k)^{\ell r} \cdot e^{-\eta m / 6} \geqs 2^{\ell r \log(2k) - \eta m / 6}$ as desired.
\end{proof}

\begin{lemma} \label{lem:random-size}
For any $0 < \alpha < 1$ and any $m \in \N$, let $S_1, \dots, S_k$ be subsets of an $m$-element universe $\cU$ drawn independently at random from $\cD_{\cU, \alpha}$. Then, $|S_1|, \dots, |S_k| \leqs 2\alpha m$ with probability at least $1 - 2^{\log k - \alpha m/3}$.
\end{lemma}

\begin{proof}
For each $i \in [k]$, since each $u \in \cU$ is included in $S_i$ independently w.p. $\alpha$, Chernoff bound implies that $\Pr[|S_i| > 2\alpha m] \leqs 2^{-\alpha m / 3}$. By union bound over all $i \in [k]$, we get the desired bound.
\end{proof}

\subsection{A Deterministic Construction} \label{app:derandomize}

\begin{proof}[Proof of Lemma~\ref{lem:well-behaved-set-deterministic}]
The existence follows immediately from Lemma~\ref{lem:random-well-behaved}. Now, one way to construct $\cS$ is to just enumerate over all choices of $S_1, \dots, S_k$. However, this is rather slow as there can be as many as $2^m$ choices of $S_i$, i.e., the running time can be as high as $\poly(m) 2^{km}$. (Saving can be made by consider only $S_i$'s of size at most $2\alpha m$, but the running time here is still (at least) $\poly(m) 2^{\alpha k m}$.) We would like to get rid of the dependency of $m$ from the exponent. We do so quite easily by dividing $\cU$ into parts of size roughly $m_0$, finding well-behaved subsets for each part and put them together. More precisely, let us consider the following algorithm.
\begin{enumerate}[(a)]
\item Partition $\cU$ into $\cU_1 \cup \dots \cup \cU_{\lfloor m / m_0 \rfloor}$ where each $\cU_i$ has size between $m_0$ and $2m_0$ (inclusive).
\item From Lemma~\ref{lem:random-well-behaved}, for each $i \in [\lfloor m / m_0 \rfloor]$ there exists a collection $\cS_i$ of $k$ subsets of $\cU_i$ such that
\begin{itemize}
\item Every subset has size at most $2\alpha |\cU_i|$.
\item $\cS_i$ is a $(\lceil \ln(2/\eta)/\alpha^\ell\rceil, \ell, \eta)$-disperser (with respect to $\cU_i)$.
\item Any subcollection $\tcS \subseteq \cS_i$ of size $\lceil 8\ln(2/\mu)/\alpha \rceil$ is $(\alpha/2, \mu)$-uniform (with respect to $\cU_i$).
\end{itemize}
Since $\cU_i$ is of size $O(m_0)$, we can enumerate all possible collections of $k$ subsets of $\cU_i$ and check whether it satisfies the three properties above to find such a collection $\cS_i$. There are $2^{O(km_0)}$ collections to enumerate over. For each collection it takes at most $\poly(m_0)2^{O(k^2)}$ time to determine whether it is a desired intersection disperser because there are no more than $(k + 1)^{2k} = 2^{O(k^2)}$ different disjoint subcollections of $k$ sets. Moreover, we can check the uniformity in $\poly(m_0)2^{O(k)}$ time. In total, this process takes time at most $\poly(m_0)2^{O(m_0k^2)}$. \label{step:brute-force}
\item Suppose that $\cS_i = \{S_{i, 1}, \dots, S_{i, k}\}$ for every $i \in [\lfloor m / m_0 \rfloor]$. We finally take $\cS = \{\bigcup_{i = 1}^{\lfloor m / m_0 \rfloor} S_{i, 1}, \dots, \bigcup_{i = 1}^{\lfloor m / m_0 \rfloor} S_{i, k}\}$
\end{enumerate}

It is not hard to see that $\cS$ satisfies the desired properties and that the algorithm runs in $\poly(m)2^{O(m_0k^2)}$ time as desired.
\end{proof}

\end{document}